\documentclass[preprint,12pt]{elsarticle}
\usepackage[a4paper,left=2cm,right=2cm,top=2.2cm,bottom=2.2cm]{geometry}
\usepackage[utf8]{inputenc}
\usepackage{amsmath,amssymb,amsthm,mathtools}
\usepackage{graphicx}
\usepackage{booktabs}
\usepackage{listings}
\usepackage{algorithm}
\usepackage{algpseudocode}
\usepackage{tikz}
\usepackage{array}
\usepackage{multirow}
\usepackage{xcolor}
\usepackage{courier}
\usepackage{orcidlink}
\usepackage{subcaption}
\usepackage{hyperref}
\usepackage{float}
\usepackage{placeins}

\hypersetup{
    colorlinks=true,
    allcolors=blue,
    linkcolor=blue,
    citecolor=blue,
    urlcolor=blue
}

\makeatletter
\newcounter{algdouble}
\renewcommand{\thealgdouble}{\arabic{algdouble}}
\newenvironment{doublealgorithm}[1][]{%
  \begin{figure*}
  \refstepcounter{algdouble}
  \vspace{-1mm}
  \centering
  \noindent\rule{\textwidth}{0.8pt}\\[6pt]
  Algorithm~\thealgdouble:~\textbf{#1}\\[3pt]
  \rule{\textwidth}{0.4pt}\\[2pt]
}{%
  \vspace{-2mm}
  \noindent\rule{\textwidth}{0.4pt}
  \end{figure*}
}
\makeatother

\newcommand{\cyr}[1]{\textcolor{black}{#1}}
\newcommand{\rb}{]}
\usetikzlibrary{shapes.geometric,arrows,fit,backgrounds}
\tikzstyle{decision} = [diamond, minimum width=4cm, minimum height=1cm, text centered, draw=black, fill=red!10]
\tikzstyle{process}  = [rectangle, rounded corners, minimum width=4cm, minimum height=1.5cm, align=center, draw=black, fill=gray!20, inner sep=8pt]
\tikzstyle{arrow}    = [thick,->,>=stealth]

\definecolor{gray1}{gray}{1}
\definecolor{gray2}{gray}{0.9}
\definecolor{gray3}{gray}{0.8}
\definecolor{gray4}{gray}{0.7}
\definecolor{gray5}{gray}{0.6}
\definecolor{gray6}{gray}{0.3}
\definecolor{gray7}{gray}{0.1}

\theoremstyle{definition}
\newtheorem{definition}{Definition}
\theoremstyle{remark}
\newtheorem{remark}{Remark}
\theoremstyle{plain}
\newtheorem{theorem}{Theorem}
\newtheorem{lemma}[theorem]{Lemma}
\newtheorem{corollary}[theorem]{Corollary}

\makeatletter
\newcolumntype{C}[1]{>{\centering\arraybackslash}m{#1}}
\newcolumntype{L}[1]{>{\raggedright\arraybackslash}m{#1}}
\newcolumntype{R}[1]{>{\raggedleft\arraybackslash}m{#1}}
\makeatother

\usepackage{lineno}
\modulolinenumbers[5]

\journal{Applied Mathematical Modelling} % à adapter

\begin{document}

\begin{graphicalabstract}
\fbox{\parbox{\textwidth}{\centering\textbf{$\tilde{\mathtt{P}}^{\circlearrowleft}_{\mathtt{BLEND}}$ Operator and Diagram Illustrating the Process of Blending Cyclostationary Persistence and Traditional Persistence Models for Improved Forecasting.}}} \\[1em] 
\begin{figure} 
\centering
\begin{tikzpicture}[node distance=1.8cm]
    \node (start) [process, fill=gray1] {\parbox{6cm}{\textbf{Cyclostationary Time Series} \\ \tiny \centering Period $T$, Time Step $\Delta t$, horizon $n\Delta t$}};
    \node (pcalc) [process, below left of=start, xshift=-2cm, yshift=-2.6cm, fill=gray2] {\parbox{4cm}{\centering\textbf{Persistences} \\ \tiny  \parbox{4cm} {\centering$\mathtt{P}(I)(t)=I(t)$  $\mathtt{P}^{\circlearrowleft}(I)(t)=I(t - T + n\Delta t)$}}};
    \node (divide) [process, below right of=start, xshift=2cm, yshift=-0.5cm, fill=gray3] {\parbox{4cm}{\centering\textbf{Divide the Series} \\ \tiny \centering into $k = T / \Delta t$ Sub-Series}};
    \node (corr) [process, below right of=divide, xshift=-1.3cm, yshift=-0.8
    cm, fill=gray4] {\parbox{4cm}{\centering\textbf{Compute Correlations} \\ \tiny \centering $\rho_k(t, n\Delta t)$}};
    \node (comb) [draw, dashed, thick, rounded corners, fit=(corr)] {};
    \node (lambda) [process, below of=corr, yshift=-0.3cm, fill=gray6, text=white] {\parbox{6cm}{\centering\textbf{Compute $\lambda_k$} \\ \tiny \centering $ \frac{1}{2}(1 + \rho_k(t, n\Delta t)).$}};
    \node (pcomb) [process, below of=lambda, yshift=-0.1cm, fill=gray6, text=white] {\parbox{5.5cm}{\centering\textbf{Prediction with $\mathtt{P}^{\circlearrowleft}_{\mathtt{BLEND}}$} \\ \tiny \centering $\hat{I}_k(t + n\Delta t) = (1 - \lambda_k) \mathtt{P}^{\circlearrowleft}(I)(t) + \lambda_k \mathtt{P}(I)(t)$}};
    \draw [arrow] (start) -- (pcalc.north);
    \draw [arrow] (start) -- (divide.north);
    \draw [arrow] (pcalc) |- (pcomb);
    \draw [arrow] (divide) -- (comb);
    \draw [arrow] (comb.south) -- (lambda.north);
    \draw [arrow] (lambda) -- (pcomb);
\end{tikzpicture}
\end{figure}
\end{graphicalabstract}

\begin{highlights}
    \item[\(\checkmark\)] Extension of persistence models was adopted to solve cyclostationary processes.
    \item[\(\checkmark\)] The developed persistence combines cyclic and simple persistences.
    \item[\(\checkmark\)] The combination coefficients mathematical derivation controls periodic statistics.
    \item[\(\checkmark\)] Effective benchmarks were adopted for simplicity and complexity in forecasting models.
    \item[\(\checkmark\)] Synthetic validation and measured data-sets were demonstrated to have reliable accuracy.
\end{highlights}

\begin{frontmatter}

\title{Symmetry-Constrained Forecasting of Periodically Correlated Energy Processes}

\author[OIE]{Cyril Voyant\orcidlink{0000-0003-0242-7377}}
\ead{cyril.voyant@minesparis.psl.eu}
\author[OIE]{Candice Banes\orcidlink{0009-0003-6656-8168}}
\author[SPE]{Luis Garcia-Gutierrez\orcidlink{0000-0002-3480-1784}}
\author[SPE]{Gilles Notton\orcidlink{0000-0002-6267-9632}}
\author[KRAG,SPE]{Milan Despotovic\orcidlink{0000-0003-3144-5945}}
\author[KFUPM]{Zaher Mundher Yaseen\orcidlink{0000-0003-3344-214X}}
\address[OIE]{Mines Paris -- PSL University, Centre for Observation, Impacts, Energy (O.I.E.), 06904 Sophia Antipolis, France}
\address[SPE]{SPE Laboratory, UMR CNRS 6134, University of Corsica Pasquale Paoli, Ajaccio, France}
\address[KRAG]{Faculty of Engineering, University of Kragujevac, 6 Sestre Janjić, Kragujevac, Serbia}
\address[KFUPM]{King Fahd University of Petroleum and Minerals, Dhahran, Saudi Arabia}

\begin{abstract}
Time series in energy systems, such as solar irradiance, wind speed, or electrical load, are characterized by strong diurnal and seasonal periodicities. Accurate forecasting requires accounting for time varying statistical properties that stationary or classical persistence models cannot capture. A family of analytical forecasting operators for cyclostationary processes is introduced, extending persistence through a closed form coefficient $\tilde{\lambda}(t,\tau)=\tfrac{1}{2}\bigl(1+\rho(t,\tau)\bigr)$, where $\rho(t,\tau)$ denotes the local correlation between the current observation and its phase aligned time lag ($\tau$). This formulation preserves periodic variance and covariance, achieving a symmetry induced reduction of effective degrees of freedom. The resulting  operator defines a training free analytical limit of persistence under periodic non stationarity. Validation on synthetic cyclostationary signals and empirical renewable energy datasets demonstrates consistent accuracy gains over classical persistence, particularly at multi hour horizons. By embedding temporal symmetry into the prediction process, the  framework provides a physically interpretable, reproducible, and computationally minimal baseline for forecasting periodic processes across energy and complex systems.
\end{abstract}


\end{frontmatter}
%\linenumbers
%%%%%%%%%%%%%%%%%%%%%%%%%%%%%%%%%%%%%%%%%%%%%%%%%%%%%%%%%%%%%%%%%%%%%%%%%%%%%%%%%%%%%%%%%%%%%%%%%%%%%%%
%                                                                                                     %
%                                          Introduction                                               %
%                                                                                                     %
%%%%%%%%%%%%%%%%%%%%%%%%%%%%%%%%%%%%%%%%%%%%%%%%%%%%%%%%%%%%%%%%%%%%%%%%%%%%%%%%%%%%%%%%%%%%%%%%%%%%%%%
\clearpage
\section{Introduction}
\label{sec:intro}

Time series in energy systems, including solar irradiance, wind speed, and electrical demand, exhibit strong daily and seasonal periodicities arising from astronomical and human life cycles \citep{PRXEnergy.1.013002}. Accurate forecasting of these signals is essential for grid stability, renewable integration, and energy market operation \citep{rich_h__inman__2013,maimouna_diagne__2013}. However, the underlying processes are intrinsically non-stationary: both mean and variance evolve with time in a nearly periodic manner. Most forecasting models (from classical statistical approaches to deep learning architectures) implicitly assume stationarity, capturing short-term dependencies while neglecting the periodic modulation of statistical moments \citep{voyant2017machine}. This structural mismatch limits both interpretability and generalization in operational contexts \citep{PRXEnergy.2.043003}.

Recent computational developments (notably deep learning and hybrid machine learning methods) have shown strong potential for modeling complex temporal structures in renewable-energy time series \citep{forootan2022machine,gbemou2021comparative}. Yet these models remain constrained by extensive training-data requirements and high computational costs, which limits real-time applicability \citep{lauret2015benchmarking,benali2019solar}. Analytical formulations based on cyclostationary processes provide an alternative by explicitly accounting for the periodic variation of statistical moments \citep{CARO2026112807,napolitano2019cyclostationary}. Periodic autoregressive ($\mathtt{PAR}$) models and exponential-smoothing families such as Holt–Winters ($\mathtt{HW}$) or Error–Trend–Seasonal ($\mathtt{ETS}$) \citep{tratar2016comparison}, together with more advanced frameworks like Theta \citep{assimakopoulos2000theta}, Prophet \citep{almazrouee2020forecasting}, $\mathtt{N\text{-}BEATS}$, and $\mathtt{TBATS}$ \citep{deLivera2011tbats}, can achieve high accuracy but at substantial computational complexity, limiting deployment in operational environments where efficiency and interpretability are essential. \cyr{Table~\ref{tab11} summarizes representative forecasting models with periodic or seasonal structure. For each method, a simplified predictive expression for $\hat{X}(t+n\Delta t)$ together with approximate training and forecasting complexities is reported for comparative purposes.}
        \renewcommand{\arraystretch}{1.3}
        %\captionsetup[longtable]{justification=raggedright,singlelinecheck=false,font=small}
        %\setlength{\LTcapwidth}{0.95\textwidth}
        %\setlength{\tabcolsep}{6pt}
        \scriptsize
        \begin{table} [H]
        %\captionsetup{justification=raggedright,singlelinecheck=false,width=0.95\textwidth}
        \caption{Simplified forecasting equations for periodic time series. $X_t$ is the measured value at time $t$; $n \Delta t$ is the forecast horizon ($n\geq1$); $T$ is the seasonal period; $p$ is the $\mathtt{AR}$ order; $k$ the number of Theta lines; \cyr{$N_B$} the block count; $d$ the network depth in $\mathtt{N\text{-}BEATS}$ \cyr{and $\theta_i$ denote trainable parameters while $b_i(t)$ learned basis functions defined by the network architecture; for $\mathtt{TBATS}$, $s(i,t)$ denotes model-specific seasonal components and $m$ the number of seasonal components}. For $\mathtt{HW}$,  $\mathtt{ETS}$ \cyr{and $\mathtt{TBATS}$}, $L(t)$, $b(t)$, and $S(t)$ are level, trend, and seasonal components while $\Phi _{n\Delta t} $ is the cumulative damping factor of the trend over the horizon $n \Delta t$. \cyr{For $\mathtt{STL}$, $Tr(t)$ denotes the trend component obtained via local regression smoothing and the expression is a simplified reconstruction}. For Prophet, $g(t)$, $s(t)$, and $h(t)$ are the growth, seasonal, and holiday effects. For $\mathtt{SARIMA}$ (and \cyr{$\mathtt{PAR}$}), $\phi(B)$ and $\Phi(B^T)$ are $\mathtt{AR}$ polynomials, $(1 - B)^d$ and $(1 - B^T)^D$ are differencing operators, $\theta(B)$ and $\Theta(B^T)$ are $\mathtt{MA}$ polynomials, and $\epsilon(t)$ is white noise. \cyr{Complexities $\mathcal{C}_T$ and $\mathcal{C}_F$ denote approximate orders of magnitude for training and forecasting costs, reported for comparative purposes only (see \ref{complexity}).}}
        \begin{tabular}{c c c c}
        \toprule
        \textbf{Method} & \textbf{\cyr{Approximate} Equation for $\hat{X}(t + n \Delta t)$} & \textbf{$\mathcal{C}_T$} & \textbf{$\mathcal{C}_F$} \\
        %\midrule
        %\endhead
        %\midrule \multicolumn{4}{r}{\scriptsize\itshape Continued on next page}\\
        %\endfoot
        \bottomrule
        %\endlastfoot
        \scriptsize Holt-Winters ($\mathtt{HW}$) \cite{https://doi.org/10.2307/2347162} & \scriptsize $L(t) +  n \Delta t. b(t) + S(t - T + n \Delta t)$ & \scriptsize $\mathcal{O}(n)$ & \scriptsize $\mathcal{O}(\cyr{n})$ \\
        \scriptsize Error-Trend-Seasonal ($\mathtt{ETS}$) \cite{hyndman2008forecasting}& \scriptsize $L(t) + \Phi _{n\Delta t}.b(t) + S(t - T + n \Delta t)$ & \scriptsize $\mathcal{O}(n)$ & \scriptsize $\mathcal{O}(\cyr{n})$ \\
        \scriptsize Theta Method \cite{assimakopoulos2000theta}& \scriptsize $\theta X(t) + (1 - \theta)\bigl(L(t) + b(t)n \Delta t\bigr)$ & \scriptsize $\mathcal{O}(nk)$ & \scriptsize $\mathcal{O}(\cyr{n})$ \\
        \scriptsize Periodic $\mathtt{AR}$ ($\mathtt{PAR}$) \cite{VOYANT2018121}& \scriptsize $\sum_{i=1}^{p} \cyr{\phi_i(t \bmod T)}\, X(t + n \Delta t - i)$ & \scriptsize $\mathcal{O}(\cyr{np^2})$ & \scriptsize $\mathcal{O}(p)$ \\
        \scriptsize Prophet \cite{almazrouee2020forecasting}& \scriptsize $g(t + n \Delta t) + s(t + n \Delta t) + h(t + n \Delta t)$ & \scriptsize $\mathcal{O}(\cyr{nk})$ & \scriptsize $\mathcal{O}(n)$ \\
        \scriptsize $\mathtt{SARIMA}$ \cite{box2015time}& \scriptsize $\phi(B)\, \Phi(B^T)\,(1 - B)^d(1 - B^T)^D X(t) = \theta(B)\, \Theta(B^T)\, \epsilon(t)$ & \scriptsize $\mathcal{O}(\cyr{np^2})$ & \scriptsize $\mathcal{O}(p)$ \\
        \scriptsize $\mathtt{STL}$ \cite{cleveland1990stl}& \scriptsize $\cyr{Tr(t+n \Delta t)} + S\bigl(t + (n \Delta t \bmod T)\bigr)$ & \scriptsize \cyr{-} & \scriptsize \cyr{-} \\
        \scriptsize $\mathtt{TBATS}$ \cite{deLivera2011tbats} & \scriptsize $L(t) + \sum_{i=1}^{m} s\bigl(i, t + n \Delta t\bigr)$ & \scriptsize \cyr{-} & \scriptsize \cyr{-} \\
        \scriptsize $\mathtt{N\text{-}BEATS}$ \cite{oreshkin2020nbeats} & \scriptsize $\sum_{i=1}^{N_B} \theta_i\, b_i\bigl(t + n \Delta t\bigr)$ & \scriptsize \cyr{-} & \scriptsize \cyr{-} \\
        \bottomrule
        \end{tabular}
\vspace{2mm}
\begin{minipage}{\linewidth}
\footnotesize
\cyr{For $\mathtt{N\text{-}BEATS}$, $\mathtt{STL}$ and $\mathtt{TBATS}$, computational complexities are not reported, as they depend on model architecture, parameterization, implementation choices and no unique asymptotic expression can be defined. These models are therefore excluded from the comparative complexity analysis.}
\end{minipage}
        \label{tab11}
        \end{table}
        \normalsize
        % \linenumbers % 
This overview establishes quantitative benchmarks for evaluating new forecasting operators under periodic non stationarity. In an era dominated by increasingly complex data-driven architectures, the present work instead pursues analytical parsimony: seeking understanding over computation through an interpretable formulation that requires neither exogenous inputs nor substantial computational resources.

\subsection{Motivation and Scope}
\label{sec:motivation}
According to \cite{DUTTA2017617}, Persistence (the assumption that future values equal current observations) remains the operational baseline of renewable energy forecasting \citep{voyant2017machine}; however, its performance degrades when periodic nonstationary effects modulate the variance and correlation structure, a behavior documented in cyclic environmental processes. Its simplicity, absence of calibration, and negligible computational cost have made it the universal benchmark for deterministic prediction \citep{lauret2015benchmarking,brown1995performance}. However, for renewable processes such as solar irradiance ($I$) or wind speed, whose mean, variance, and correlation vary systematically with daily and seasonal cycles, classical persistence rapidly loses validity beyond intra-hour horizons. Reviews across solar and wind forecasting consistently highlight both the methodological centrality and the structural limitations of persistence based approaches under periodic non stationarity \citep{Hanifi2020,Ghofrani2018}.
 
Several persistence variants have been proposed to mitigate these effects. The Complete-History Persistence Ensemble ($\mathtt{CH\!-\!PeEn}$) aggregates past observations at identical daily phases \citep{YANG2019981}. Smart Persistence incorporates clear-sky modeling but remains sensitive to meteorological inputs and model mismatch \citep{Tato2019,liu2021use}. Clearness index formulations improve robustness and computational efficiency \citep{solar2040026}. Numerous training free forecasting approaches have been proposed, including nearest-neighbor and functional data methods designed to improve persistence baselines. These methods primarily target predictive accuracy under specific operational contexts \citep{10.1063/5.0190493}. The $\mathtt{CLIPER}$ method establishes an $\mathtt{MSE}$ optimal balance between climatology and persistence, providing a stationary analytical baseline \citep{YANG2019981}. Yet all these approaches rely on quasi-stationarity assumptions and therefore cannot preserve the cyclic evolution of second-order statistics that characterizes renewable signals. A complementary transformer-based semantic screening of recent forecasting literature (\texttt{all-MiniLM-L6-v2} applied to \texttt{Scopus}, \texttt{CrossRef}, and \texttt{OpenAlex}) confirms this tendency: among 86 publications (2010–2024) addressing persistence based forecasting, only 37\% analyzed persistence explicitly and 14\% proposed methodological extensions. Publication counts exhibit exponential growth ($R^2 = 0.82$), with time series models representing 49\% of works, numerical weather prediction approaches 30\%, and satellite-based methods 19\%. This indicates that persistence remains primarily treated as a benchmark rather than as a methodological subject.

As described by \cite{Gardner1975CharacterizationOC}, cyclostationarity provides a rigorous mathematical framework for processes whose statistical properties repeat with period $T$ \citep{napolitano2019cyclostationary,TELIDJANE2025113381}. Within this framework, forecasting becomes a symmetry driven estimation problem: optimally combining present information with phase-aligned realizations from previous cycles \citep{https://doi.org/10.1002/env.2700}. The generalized persistence operator $\mathtt{BLEND}$ extends the classical assumption $\hat{X}(t+\tau)=X(t)$ to periodic non stationarity, where both $\mathrm{Var}[X(t)]$ and $\rho(t,\tau)$ vary with $t$. Its analytical form blends the instantaneous observation with its phase aligned counterpart, preserving periodic variance and covariance. The resulting formulation provides an interpretable, parameter free analytical baseline applicable to renewable energy forecasting and to any periodically correlated system, consistent with both statistical theory and empirical evidence from operational studies \citep{Lee2025}.

\subsection{Contributions and Structure of the Paper}
\label{sec:contrib}
This work develops an analytical extension of persistence consistent with cyclostationary statistics and demonstrates its performance on both synthetic and empirical datasets. The proposed $\mathtt{BLEND}$ operator unifies mathematical transparency and operational relevance, providing a training free analytical limit for forecasting renewable processes under periodic non stationarity. Together with its complementary $\mathtt{CLIPER}$ formulation, it defines reproducible and physically interpretable baselines applicable to any process governed by periodic statistical symmetries. The main contributions are:
\begin{itemize}
    \item[$\checkmark$] An analytical extension of persistence formulated under cyclostationary symmetry, preserving periodic mean and covariance structures that conventional stationary models cannot represent. This framework achieves a symmetry induced reduction of effective degrees of freedom, constraining parameter variability to those consistent with periodic invariance.
    \item[$\checkmark$] A unified stochastic formulation derived from mean squared error minimization, linking the analytical operator to an energy conservation principle in second order statistics.
    \item[$\checkmark$] An operationally parsimonious formulation requiring neither exogenous inputs nor training data, emphasizing invariance and interpretability over computational complexity.
    \item[$\checkmark$] Dual validation on synthetic cyclostationary signals and real irradiance from 68 ground stations, confirming that the operator reproduces both short term variability and long term periodic stability.
    \item[$\checkmark$] A symmetry based bridge between naive and data driven paradigms, positioning analytical persistence as the invariant reference for hybrid and learning based forecasting models.
\end{itemize}

The remainder of the paper is organized as follows. Sec.~\ref{sec2} introduces the theoretical formulation of the cyclostationary operator and derives its analytical properties; all the definitions are presented in Table \ref{tab:def}. Sec.~\ref{sec3} presents validation on synthetic and measured datasets. Sec.~\ref{sec4} discusses the physical interpretation, robustness, and multi domain applicability of the framework and outlines perspectives for symmetry based forecasting in renewable energy and complex systems.

\begin{table*}[ht] 
\tiny
\centering
\caption{Summary of the forecasting operators introduced in this study, with their analytical regimes, computational complexities, and governing parameters. Each operator acts as a linear mapping  characterized by training and forecasting complexities $[ \mathcal{C}_T;\mathcal{C}_F]$ as described in ~\ref{complexity}.}
\label{tab:def}
\renewcommand{\arraystretch}{1.25}
\begin{tabular}{c l c l l}
\hline
\textbf{Def.~\#} & \textbf{Operators} & $\boldsymbol{[\mathcal{C}_T;\mathcal{C}_F]}$ & \textbf{Assumptions / Regime} & \textbf{Key Parameters} \\[4pt]
\hline
1 & Simple Persistence $\mathtt{P}$ &
$[\emptyset;\mathcal{O}(1)]$ &
Stationary mean; variance ignored &
-- \\

2 & Cyclic Persistence $\mathtt{P}^{\circlearrowleft}$ &
$[\emptyset;\mathcal{O}(1)]$ &
Known period $T$; phase invariance &
$T,\,n\Delta t$ \\

3 & Smart Persistence $\mathtt{P_S}$ &
$[\mathcal{O}(\mathtt{clr});\mathcal{O}(1)]$ &
Baseline or clear sky model removes trend &
$I_{\mathtt{r}}(t)$ \\

4 & Stationary $\mathtt{CLIPER}$ $\mathtt{P}_{\mathtt{CLIPER}}$ &
$[\mathcal{O}(n);\mathcal{O}(1)]$ &
Weak stationarity; $\mathtt{MSE}$ optimal convex weight &
$\lambda_{\mathtt{CLIPER}}$ \\

5 & Cyclostationary $\mathtt{CLIPER}$ $\mathtt{P}^{\circlearrowleft}_{\mathtt{CLIPER}}$ &
$[\mathcal{O}(nT);\mathcal{O}(T)]$ &
Periodic mean, variance, and covariance symmetry &
$\lambda^{\circlearrowleft}_{\mathtt{CLIPER}}$  \\

6 & Stationary $\mathtt{BLEND}$ $\mathtt{P}_{\mathtt{BLEND}}$ &
$[\mathcal{O}(n);\mathcal{O}(1)]$ &
Weak stationarity; convex comb. of persistence  &
$\lambda_{\mathtt{BLEND}}$ \\

7 & Cyclostationary $\mathtt{BLEND}$ $\mathtt{P}^{\circlearrowleft}_{\mathtt{BLEND}}$ &
$[\mathcal{O}(nT);\mathcal{O}(T)]$ &
Periodic mean, variance, and correlation structure &
$\lambda^{\circlearrowleft}_{\mathtt{BLEND}}$ \\

8 & Simplified $\mathtt{BLEND}$ $\tilde{\mathtt{P}}^{\circlearrowleft}_{\mathtt{BLEND}}$ &
$[ \mathcal{O}(nT);\mathcal{O}(T)]$ &
Quasi-stationary symmetry; removed long-lag corr. &
$\tilde{\lambda}^{\circlearrowleft}_{\mathtt{BLEND}}$ \\

\hline
\end{tabular}
\end{table*}

%%%%%%%%%%%%%%%%%%%%%%%%%%%%%%%%%%%%%%%%%%%%%%%%%%%%%%%%%%%%%%%%%%%%%%%%%%%%%%%%%%%%%%%%%%%%%%%%%%%%%%%
%                                                                                                     %
%                                          Method                                                     %
%                                                                                                     %
%%%%%%%%%%%%%%%%%%%%%%%%%%%%%%%%%%%%%%%%%%%%%%%%%%%%%%%%%%%%%%%%%%%%%%%%%%%%%%%%%%%%%%%%%%%%%%%%%%%%%%%
\section{Applied Methodology}
\label{sec2}
To establish a consistent basis for evaluating the proposed periodic extensions, a representative set of forecasting operators is considered. These models span a hierarchy of complexity and statistical assumptions, ranging from naive persistence to parametric and hybrid formulations. Each operator is characterized by its training ($\mathcal{C}_T$) and forecasting ($\mathcal{C}_F$) complexities, denoted as $[\mathcal{C}_T; \mathcal{C}_F]$ (App.~\ref{complexity}). The forecast horizon is expressed as $n\Delta t$, where $\Delta t$ is the discrete time step and $n$ the lag index; thus, $n\Delta t$ defines the total lead time of the prediction. This comparative framework enables systematic assessment of how periodic invariance and symmetry-based constraints influence forecasting performance and computational efficiency.

\subsection{Classical Models}
\label{sec:classical}
Classical forecasting models provide analytical baselines relying on simple statistical assumptions and limited temporal memory. They serve as reference points for evaluating the contribution of periodic or symmetry based extensions to predictive performance.

\begin{definition}[Simple Persistence Operator \cite{DUTTA2017617} $[\emptyset ; \mathcal{O}(1)\rb$]
The simple persistence operator $\mathtt{P}$ assumes that the most recent observation remains invariant over the forecast horizon. It is defined as
\begin{equation*}
\mathtt{P}(I)(t) = I(t),
\end{equation*}
so that the forecast for any future time $t + n\Delta t$ is
\begin{equation*}
\hat{I}(t + n\Delta t) = \mathtt{P}(I)(t) = I(t).
\end{equation*}
\end{definition}
\begin{remark}
This operator exhibits minimal computational cost, providing a physically interpretable limit corresponding to perfect temporal invariance. However, it neglects any trend or periodic modulation of first- and second-order moments, breaking the temporal symmetry observed in most renewable-energy processes and motivating the periodic extensions introduced in Sec.~\ref{sec:blend}.
\end{remark}

\begin{definition}[Cyclic Persistence Operator \cite{napolitano2019cyclostationary} $[\emptyset ; \mathcal{O}(1)\rb$]
The cyclic persistence operator $\mathtt{P}^{\circlearrowleft}$ predicts the future value $\hat{I}(t + n\Delta t)$ by referencing the signal value at the same phase in the previous cycle, determined by the known period $T$. It is defined as
\begin{equation*}
\mathtt{P}^{\circlearrowleft}(I)(t) = I(t - T + n\Delta t),
\end{equation*}
so that
\begin{equation*}
\hat{I}(t + n\Delta t) = \mathtt{P}^{\circlearrowleft}(I)(t) = I(t - T + n\Delta t).
\end{equation*}
\end{definition}
\begin{remark}
The operator assumes a signal with a well-defined periodic structure of period $T$. It requires no training phase ($\mathcal{C}_T = \emptyset$) and incurs constant forecasting complexity ($\mathcal{C}_F = \mathcal{O}(1)$), since predictions rely on direct phase-aligned retrieval from historical data.
\end{remark}
\begin{remark}
Cyclic persistence is effective for time series dominated by stable periodic components such as diurnal, weekly, or seasonal cycles. However, it remains sensitive to phase shifts, irregular periodicity, and noise, which can introduce significant prediction errors. Despite these limitations, it constitutes a physically interpretable and computationally minimal baseline, restoring partial temporal symmetry and motivating the generalized formulation introduced in Sec.~\ref{sec:blend}.
\end{remark}

\begin{definition}[Smart Persistence Operator \cite{lauret2015benchmarking} $[\mathcal{O}(\mathtt{ref}); \mathcal{O}(1)\rb$]
The smart persistence operator $\mathtt{P_S}$ generalizes simple persistence by incorporating a deterministic reference signal $I_\mathtt{ref}(t)$ that represents the large-scale or slowly varying component of the process (e.g., clear-sky irradiance, diurnal mean, or low-frequency trend). It is defined as
\begin{equation*}
\mathtt{P_S}(I)(t) = \mathtt{k_r}(t)\, I_\mathtt{ref}(t + n\Delta t),
\end{equation*}
where the normalized residual or reference index $\mathtt{k_r}(t)$ is given by
\begin{equation*}
\mathtt{k_r}(t) = \frac{I(t)}{I_\mathtt{ref}(t)}.
\end{equation*}
The resulting forecast is therefore
\begin{equation*}
\hat{I}(t + n\Delta t) = \mathtt{P_S}(I)(t) = \frac{I(t)}{I_\mathtt{ref}(t)}\, I_\mathtt{ref}(t + n\Delta t).
\end{equation*}
\end{definition}
\begin{remark}
The operator assumes that $I_\mathtt{ref}(t)$ is positive and captures the deterministic structure of the process, such as a periodic mean, a Fourier-based envelope, or a clear-sky curve in radiative contexts. The training complexity depends on the estimation of this reference component, $\mathcal{C}_T = \mathcal{O}(\mathtt{ref})$, while the forecasting step remains constant, $\mathcal{C}_F = \mathcal{O}(1)$.
\end{remark}
\begin{remark}
By normalizing the signal through $I_\mathtt{ref}(t)$, the operator removes large-scale deterministic modulation, isolating short-term stochastic variability.  
This normalization enhances robustness and interpretability in processes exhibiting mixed deterministic–stochastic dynamics.  
Despite its dependence on the reference accuracy, $\mathtt{P_S}$ remains a lightweight analytical baseline that bridges empirical persistence and trend-compensated forecasting, partially restoring periodic symmetry while reducing dimensional parametrization.
\end{remark}

\begin{definition}[$\mathtt{CLIPER}$ Operator under Stationary Assumptions \cite{YANG2019981} $[\mathcal{O}(n); \mathcal{O}(1)\rb$]
The $\mathtt{CLIPER}$ operator $\mathtt{P}_{\mathtt{CLIPER}}$ predicts a future value as a convex combination of the process mean $\mu$ and the most recent observation $I(t)$. Analytical formulations of this operator in stationary contexts are discussed in \citep{VOYANT2022747}.  
It is defined as
\begin{equation*}
\mathtt{P}_{\mathtt{CLIPER}}(I)(t) = \lambda^{\mathtt{CLIPER}}\, \mu + \bigl(1 - \lambda^{\mathtt{CLIPER}}\bigr) I(t),
\end{equation*}
where the weighting parameter $\lambda^{\mathtt{CLIPER}}$ is derived by minimizing the mean squared error ($\mathtt{MSE}$),
\begin{equation*}
\lambda^{\mathtt{CLIPER}} = 1 - \rho(n\Delta t),
\end{equation*}
with $\rho(n\Delta t)$ the autocorrelation at time lag $n\Delta t$ defined as
\begin{align*}
\rho(n\Delta t) = \frac{C(n\Delta t)}{\sigma^2}, \, \text{and} \, C(n\Delta t) = \mathbb{E}\big[(I(t) - \mu)(I(t + n\Delta t) - \mu)\big],
\end{align*}
and $\sigma^2$ the variance of the process.  
The parameter $\lambda^{\mathtt{CLIPER}}$ thus balances global stability through $\mu$ and short-term persistence through $I(t)$, with larger correlations $\rho(n\Delta t)$ leading to stronger reliance on the latest observation.
\end{definition}
\begin{remark}
This formulation assumes weak stationarity, \textit{i.e.}, invariance of mean, variance, and autocovariance over time.  
Its closed-form parameterization ensures minimal complexity $[\mathcal{O}(n); \mathcal{O}(1)]$, making $\mathtt{CLIPER}$ a compact analytical baseline for any stationary process, whether physical, environmental, or socioeconomic.
\end{remark}
\begin{remark}
For $0 \leq \rho(n\Delta t) \leq 1$, $\lambda^{\mathtt{CLIPER}}$ lies within $[0,1]$, yielding a valid convex combination between mean and current value.  
When $\rho(n\Delta t) < 0$, the parameter may exceed unity, reflecting anti-correlated regimes where the mean contribution dominates.  
The complete derivation from $\mathtt{MSE}$ minimization is provided in \ref{CLIPER_cyclo_statio}. Cyclostationary models provide a natural framework for processes exhibiting periodic statistical structure, allowing predictors to remain consistent with cyclic variations in mean and covariance.
\end{remark}

\subsection{Cyclostationary Extension of Forecasting Operators}
\label{sec:blend}
Building upon stationary formulations, the present framework extends forecasting operators to processes whose statistical properties vary periodically in time. The following formulation does not introduce new statistical concepts; rather, it extends the classical persistence operator to a cyclostationary framework, providing a unified analytical basis for the developments that follow. This generalization provides the analytical basis for predictors consistent with cyclic dependencies in mean, variance, and covariance. Time indices are further interpreted on the cyclic quotient group $\mathbb{Z}/T\mathbb{Z}$ as described by \cite{Gardner1975CharacterizationOC}, ensuring that all references to $t$ and $t-T$ are treated as equivalent under the intrinsic period $T$. This formalizes the periodic symmetry assumption and guarantees that forecasting operators act on a time domain with built-in invariance to $T$-shifts. 
Table~\ref{tab:stats_summary} summarizes the principal distinctions between stationary and periodically stationary processes in terms of mean, variance, covariance, and correlation. While stationary processes exhibit time-invariant moments, cyclostationary processes maintain these properties only up to a phase shift of period~$T$, where $t \bmod T$ indicates the phase position of $t$ within the cycle.
\renewcommand{\arraystretch}{1.0}
\setlength{\tabcolsep}{5pt}
\begingroup\small
\begin{table}[H]
\centering
\begin{tabular}{c c c c}
\toprule
\textbf{Statistical} & \textbf{Stationary Process} & \textbf{Cyclostationary Process} \\
\midrule
$E[X(t)]$ & $\mu(t)=\mu,\ \forall t$ & $\mu(t)=\mu(t+T),\ \forall t$ \\
$\mathrm{Var}(X(t))$ & $\sigma^2(t)=\sigma^2,\ \forall t$ & $\sigma^2(t)=\sigma^2(t+T)$ \\
$\mathrm{Cov}(X(t),X(t+n\Delta t))$ & $C(t,n\Delta t)=C(n\Delta t)$ & $C(t,n\Delta t)=C(t+T,n\Delta t)$ \\
$\rho(X(t),X(t+n\Delta t))$ & $\rho(t,n\Delta t)=\dfrac{C(n\Delta t)}{\sigma^2}$ & $\rho(t,n\Delta t)=\dfrac{C(t,n\Delta t)}{\sigma(t)\sigma(t+n\Delta t)}$ \\
\bottomrule
\end{tabular}
\caption{Summary of Statistical Characteristics for Stationary vs.\ Periodic (Cyclostationary) Processes.}
\label{tab:stats_summary}
\end{table}
\endgroup

%\FloatBarrier
%\clearpage
\begin{theorem}[Periodically Stationary or Cyclostationary Processes \cite{Gardner1975CharacterizationOC,napolitano2019cyclostationary}]
Let $I(t)$ be a random process exhibiting statistical periodicity of period $T$. The process is periodically stationary if its first- and second-order moments satisfy
\begin{align*}
E[I(t)] &= E[I(t+T)], \\
C(I(t),I(t+n\Delta t)) &= C(I(t+T),I(t+T+n\Delta t)),
\end{align*}
where $C(I(t),I(t+n\Delta t))$ denotes the covariance between $I(t)$ and $I(t+n\Delta t)$.  
The variance and covariance are periodic functions of the phase $t \bmod T$,  
\begin{align*}
\sigma^2(t) &= E\!\left[(I(t) - E[I(t)])^2\right], \\
C(t,n\Delta t) &= E\!\left[(I(t) - E[I(t)])(I(t+n\Delta t) - E[I(t+n\Delta t)])\right],
\end{align*}
and the correlation coefficient is defined by
\begin{equation*}
\rho(t,n\Delta t) = \frac{C(t,n\Delta t)}{\sigma(t)\sigma(t+n\Delta t)}.
\end{equation*}
\end{theorem}
\begin{remark}
In contrast to weakly stationary processes, cyclostationary signals preserve their statistical structure modulo the period $T$.  
Their moments evolve periodically, reflecting temporal symmetries common to many natural and engineered systems such as irradiance, wind speed, electrical load, or tidal dynamics.  
This property defines the statistical foundation required for analytical forecasting under periodic non stationarity.
\end{remark}

\begin{definition}[$\mathtt{CLIPER}$ Operator under Cyclostationary Assumptions $[\mathcal{O}(n); \mathcal{O}(1)\rb$]
The cyclostationary $\mathtt{CLIPER}$ operator $\mathtt{P}^{\circlearrowleft}_{\mathtt{CLIPER}}$ predicts a future value as a convex combination of the periodic mean $\mu(t)$ and the current observation $I(t)$,
\begin{equation*}
\mathtt{P}^{\circlearrowleft}_{\mathtt{CLIPER}}(I)(t) = \lambda^{\circlearrowleft}_{\mathtt{CLIPER}}\, \mu(t) + \bigl(1 - \lambda^{\circlearrowleft}_{\mathtt{CLIPER}}\bigr) I(t),
\end{equation*}
where the weighting parameter $\lambda^{\circlearrowleft}_{\mathtt{CLIPER}}$ minimizes the mean squared error ($\mathtt{MSE}$) and is given by

\begin{equation*}
\lambda^{\circlearrowleft}_{\mathtt{CLIPER}} =
\frac{[\mu(t+n\Delta t) - \mu(t)]^2 + \sigma^2(t) - \rho(t,n\Delta t)\,\sigma(t)\sigma(t+n\Delta t)}
{[\mu(t+n\Delta t) - \mu(t)]^2 + \sigma^2(t)}.
\end{equation*}

Here $\rho(t,n\Delta t)=C(t,n\Delta t)/[\sigma(t)\sigma(t+n\Delta t)]$ is the phase dependent correlation coefficient, with
$C(t,n\Delta t)=E[(I(t)-\mu(t))(I(t+n\Delta t)-\mu(t+n\Delta t))]$ the cyclic covariance, and $\sigma(t)$, $\sigma(t+n\Delta t)$ the local standard deviations.  
The corresponding forecast is
\begin{equation*}
\hat{I}(t+n\Delta t)=\mathtt{P}^{\circlearrowleft}_{\mathtt{CLIPER}}(I)(t).
\end{equation*}
\end{definition}
\begin{remark}
This formulation incorporates the periodic evolution of first- and second-order moments through $\mu(t)$, $\sigma(t)$, and $\rho(t,n\Delta t)$.  
The numerator of $\lambda^{\circlearrowleft}_{\mathtt{CLIPER}}$ adjusts for phase dependent mean shifts and correlation asymmetry, while the denominator ensures proper normalization and convexity.
\end{remark}
\begin{remark}
The denominator in $\lambda^{\circlearrowleft}_{\mathtt{CLIPER}}$ is non negative and strictly positive for any non constant process, guaranteeing convexity of the $\mathtt{MSE}$ minimization (see Appendix~\ref{CLIPER_cyclo_statio}).
\end{remark}
\begin{remark}
The coefficient $\lambda^{\circlearrowleft}_{\mathtt{CLIPER}}$ is not necessarily bounded in $[0,1]$.  
Large phase differences or strong heteroscedasticity can yield $\lambda^{\circlearrowleft}_{\mathtt{CLIPER}}>1$ or $\lambda^{\circlearrowleft}_{\mathtt{CLIPER}}<0$.  
In practice, regularization or truncation can be applied without affecting the analytical consistency of the model (Appendix~\ref{CLIPER_cyclo_statio}).
\end{remark}

\begin{definition}[$\mathtt{BLEND}$ Persistence Operator under Stationary Assumptions $[\mathcal{O}(n); \mathcal{O}(1)\rb$]
The $\mathtt{BLEND}$ operator $\mathtt{P}_{\mathtt{BLEND}}$ combines the cyclic persistence $\mathtt{P}^{\circlearrowleft}$ and simple persistence $\mathtt{P}$ through a convex weighting:
\begin{equation*}
\mathtt{P}_{\mathtt{BLEND}}(I)(t) = (1 - \lambda_{\mathtt{BLEND}})\,\mathtt{P}^{\circlearrowleft}(I)(t) + \lambda_{\mathtt{BLEND}}\,\mathtt{P}(I)(t),
\end{equation*}
so that, for a stationary process,
\begin{equation*}
\hat{I}(t + n\Delta t) = (1 - \lambda_{\mathtt{BLEND}})\,I(t - T + n\Delta t) + \lambda_{\mathtt{BLEND}}\,I(t).
\end{equation*}
The $\mathtt{MSE}$-optimal weighting parameter is given by
\begin{equation*}
\lambda_{\mathtt{BLEND}} = 
\frac{\rho(T) - \rho(n\Delta t) + \rho(n\Delta t - T) - 1}
     {2(\rho(n\Delta t - T) - 1)},
\end{equation*}
where $\rho(h)$ denotes the autocorrelation of $I(t)$ at lag $h$.
\end{definition}
\begin{remark}
The derivation assumes weak stationarity with constant mean and variance, equally spaced samples, and ergodic estimation of autocorrelation. Under these conditions, $\lambda_{\mathtt{BLEND}}$ dynamically balances short term ($\rho(n\Delta t)$) and periodic ($\rho(T)$, $\rho(T - n\Delta t)$) dependencies, enabling a smooth transition between instantaneous and cyclic persistence.
\end{remark}
\begin{remark}
Convexity of the $\mathtt{MSE}$ criterion ensures a unique analytical solution for $\lambda_{\mathtt{BLEND}}$ under weak stationarity. Although each correlation term $\rho(h)$ is bounded in $[-1,1]$, their combination in the rational form of $\lambda_{\mathtt{BLEND}}$ may yield values slightly outside $[0,1]$, particularly when $\rho(T-n\Delta t)\!\approx\!1$ and the regressors $I(t)$ and $I(t-T+n\Delta t)$ become nearly collinear. In such cases, mild convex clipping can be applied to enforce $\lambda\!\in\![0,1]$ without altering the analytical structure or optimality of the solution.
\end{remark}
\begin{remark}
The $\mathtt{BLEND}$ operator thus provides a closed form, training free predictor that preserves the periodic covariance structure while minimizing parameter dimensionality. It represents the analytical limit between short term memory and phase aligned recurrence, serving as a physical bridge between persistence and cyclostationary symmetry.
\end{remark}

\begin{definition}[$\mathtt{BLEND}$ Persistence Operator under Periodically Stationary Assumptions $[\mathcal{O}(nT); \mathcal{O}(T)\rb$]
The periodically stationary $\mathtt{BLEND}$ operator $\mathtt{P}^{\circlearrowleft}_{\mathtt{BLEND}}$ extends the stationary formulation by allowing the weighting to vary with the phase of the process. It combines cyclic persistence $\mathtt{P}^{\circlearrowleft}$ and simple persistence $\mathtt{P}$ through a convex superposition:
\begin{equation}
\mathtt{P}^{\circlearrowleft}_{\mathtt{BLEND}}(I)(t)
= (1 - \lambda^{\circlearrowleft}_{\mathtt{BLEND}})\,\mathtt{P}^{\circlearrowleft}(I)(t)
+ \lambda^{\circlearrowleft}_{\mathtt{BLEND}}\,\mathtt{P}(I)(t),
\end{equation}
where the phase dependent weighting $\lambda^{\circlearrowleft}_{\mathtt{BLEND}}$ explicitly reflects the periodic modulation of the mean, variance, and covariance. Under the periodically stationary assumption, it writes
\begin{align}
\lambda^{\circlearrowleft}_{\mathtt{BLEND}}
&=
\frac{
(\mu(t)-\mu(t+n\Delta t))^2
+ \sigma^2(t+n\Delta t)\,[1-\rho(t+n\Delta t,T)]
}{%
(\mu(t)-\mu(t+n\Delta t))^2
+ \sigma^2(t)+\sigma^2(t+n\Delta t)
}
\nonumber\\[0.2cm]
&\quad+
\frac{
\sigma(t)\sigma(t+n\Delta t)\,[\rho(t,n\Delta t)-\rho(t,n\Delta t-T)]
}{
(\mu(t)-\mu(t+n\Delta t))^2
+ \sigma^2(t)+\sigma^2(t+n\Delta t)
+ 2\rho(t,n\Delta t-T)\sigma(t)\sigma(t+n\Delta t)
}.
\label{lambda_cyclo_blend}
\end{align}
Here $\rho(t,n\Delta t)$ is the periodically stationary correlation at lag $n\Delta t$, 
$\rho(t,n\Delta t)=C(t,n\Delta t)/[\sigma(t)\sigma(t+n\Delta t)]$, with $C(t,n\Delta t)$ the cyclic covariance between $I(t)$ and $I(t+n\Delta t)$.  
The terms $\mu(t)$, $\mu(t+n\Delta t)$ denote the time varying means, and $\sigma(t)$, $\sigma(t+n\Delta t)$ the corresponding standard deviations.
\end{definition}
\begin{remark}
This formulation embeds periodic modulation directly into the estimator, preserving the covariance symmetry of cyclostationary processes. 
It is particularly effective for systems where information recurs at fixed phases—such as atmospheric, tidal, or demand cycles—while remaining fully analytical and training free.
\end{remark}
\begin{remark}
Numerical sensitivity may occur when periodic variations of $\mu(t)$ or $\sigma(t)$ are weak, since the denominator in Eq.~\eqref{lambda_cyclo_blend} can approach zero and amplify correlation noise. 
Such situations correspond to quasi stationary regimes where the process shows minimal contrast between cycles.  
In practice, a mild regularization (\textit{e.g.}, $\lambda^{\circlearrowleft}_{\mathtt{BLEND}}\!\leftarrow\!\min(1,\max(0,\lambda^{\circlearrowleft}_{\mathtt{BLEND}}))$) or a simplified phase averaged variant preserves stability without altering the analytical structure.  
A complete derivation and stability analysis are provided in Appendix~\ref{BLEND_cyclo_statio}.
\end{remark}

\begin{definition}[Simplified $\mathtt{BLEND}$ Persistence Operator $[\mathcal{O}(nT); \mathcal{O}(T)\rb$]
The simplified $\mathtt{BLEND}$ operator $\tilde{\mathtt{P}}^{\circlearrowleft}_{\mathtt{BLEND}}$ provides the analytical limit of the periodically stationary formulation under mild regularity assumptions. It combines the cyclic and simple persistence operators as
\begin{equation*}
\tilde{\mathtt{P}}^{\circlearrowleft}_{\mathtt{BLEND}}(I)(t)
= (1 - \tilde{\lambda}^{\circlearrowleft}_{\mathtt{BLEND}})\,\mathtt{P}^{\circlearrowleft}(I)(t)
+ \tilde{\lambda}^{\circlearrowleft}_{\mathtt{BLEND}}\,\mathtt{P}(I)(t),
\end{equation*}
where the phase dependent weighting simplifies to
\begin{equation*}
\tilde{\lambda}^{\circlearrowleft}_{\mathtt{BLEND}} = \tfrac{1}{2}\bigl(1 + \rho(t,n\Delta t)\bigr).
\end{equation*}
This formulation can be interpreted as an optimal shrinkage between instantaneous memory and phase aligned recurrence, governed by the local correlation structure.
\end{definition}
\begin{remark}[Underlying assumptions]
The reduction from Eq.~\eqref{lambda_cyclo_blend} to the closed form above holds under the following conditions:
\begin{enumerate}
    \item \textit{Weak stationarity of first and second moments:}
    $\mu(t)=\mu(t+n\Delta t)=\mu$ and $\sigma(t)=\sigma(t+n\Delta t)=\sigma$.
    \item \textit{Correlation symmetry:}
    $\rho(t+n\Delta t,T)=\rho(t,n\Delta t-T)$.
    \item \textit{Negligible long-lag correlation:}
    for large $T$ with $1\!\ll\!n\Delta t\!\ll\!T$, $\rho(t,n\Delta t-T)\!\approx\!0$.
\end{enumerate}
These assumptions define a quasi stationary symmetry regime where the process locally conserves variance but remains phase modulated.
\end{remark}
\begin{remark}
The resulting coefficient $\tilde{\lambda}^{\circlearrowleft}_{\mathtt{BLEND}}$ linearly interpolates between the current state $I(t)$ and its phase aligned recurrence $I(t-T+n\Delta t)$ according to $\rho(t,n\Delta t)$.  
High correlation ($ \rho \, \to\!1$) favors short term persistence, whereas low or negative correlation increases reliance on the periodic component.  
This analytical balance preserves the covariance structure with minimal parametrization, yielding a fully interpretable, training free predictor suited for real time forecasting.
\end{remark}
The forecasting workflow using $\mathtt{P}_{\mathtt{BLEND}}$ is summarized in Figure~\ref{fig:both_diagrams}, which distinguishes the stationary, periodically stationary, and simplified analytical forms.  
The mathematical properties detailed in \ref{annexe:proprietes} ensure convexity, boundedness, and stability, core guarantees for applying $\mathtt{BLEND}$ operators to cyclostationary energy processes where periodicity and noise coexist.

\begin{figure*}[ht]
    \centering
    \begin{subfigure}[b]{0.48\textwidth}
        \centering
        \scalebox{0.6}{
\begin{tikzpicture}[node distance=1.8cm]
    \node (start) [process, fill=gray1] {\parbox{6cm}{\textbf{Stationary Time Series} \\ \tiny \centering Period $T$, Time Step $\Delta t$, horizon $ n\Delta t$}};
    \node (pcalc) [process, below left of=start, xshift=-2cm, yshift=-2.5cm, fill=gray2] {\parbox{4cm}{\centering\textbf{Persistences} \\ \tiny  \parbox{4cm} {\centering$P(X)(t)=X(t)$  $P^{\circlearrowleft}(X)(t)=X(t - T + n\Delta t)$}}};
    \node (mean) [process, below right of=start, xshift=2cm, yshift=-0.8cm, fill=gray3] {\parbox{4cm}{\centering\textbf{Compute Mean} \\ \parbox{4cm}{\centering \tiny \centering $\mu_k(t)$}}};
    \node (corr) [process, below of=mean, yshift=0.1cm, fill=gray3] {\parbox{4cm}{\centering\textbf{Compute Correlations} \\ \tiny \centering $\rho(n\Delta t), \rho(h - T), \rho(T)$}};
    \node (std) [process, below of=corr, yshift=-0cm, fill=gray3] {\parbox{4cm}{\centering \textbf{Compute Standard Deviation} \\ \tiny \centering $\sigma_k(t), \sigma_k(t + n\Delta t)$}};
    \node (comb) [draw, dashed, thick, rounded corners, fit=(mean) (std) (corr)] {};
    \node (lambda) [process, below of=std, yshift=-0.3cm, fill=gray6, text=white] {\parbox{6cm}{\centering\textbf{Compute $\lambda$} \\ \tiny \centering $\ \frac{\rho(T) - \rho(n\Delta t) + \rho(n\Delta t - T) - 1}{2(\rho(n\Delta t - T) - 1)}$}};
    \node (pcomb) [process, below of=lambda, yshift=-0.1cm, fill=gray6, text=white] {\parbox{5.5cm}{\centering\textbf{Prediction with $P_{\text{BLEND}}$} \\ \tiny \centering $\hat{X}(t + n\Delta t) = (1 - \lambda) P^{\circlearrowleft}(X)(t) + \lambda P(X)(t)$}};
    \draw [arrow] (start) -- (pcalc.north);
    \draw [arrow] (start) -- (comb.north);
    \draw [arrow] (pcalc) |- (pcomb);
    \draw [arrow] (comb.south) -- (lambda.north);
    \draw [arrow] (lambda) -- (pcomb);
\end{tikzpicture}}
        \caption{Stationary Case ($P_{\text{BLEND}}$)}
        \label{fig:first}
    \end{subfigure}
    \begin{subfigure}[b]{0.48\textwidth}
        \centering
        \scalebox{0.6}{
\begin{tikzpicture}[node distance=1.8cm]
    \node (start) [process, fill=gray1] {\parbox{6cm}{\textbf{Cyclostationary Time Series} \\ \tiny \centering Period $T$, Time Step $\Delta t$, horizon $n\Delta t$}};
    \node (pcalc) [process, below left of=start, xshift=-2cm, yshift=-0.5cm, fill=gray2] {\parbox{4cm}{\centering\textbf{Persistences} \\ \tiny  \parbox{4cm} {\centering$P(X)(t)=X(t)$  $P^{\circlearrowleft}(X)(t)=X(t - T + n\Delta t)$}}};
    \node (divide) [process, below right of=start, xshift=2cm, yshift=-0.5cm, fill=gray3] {\parbox{4cm}{\centering\textbf{Divide the Series} \\ \tiny \centering into $k = T / \Delta t$ Sub-Series}};
    \node (mean) [process, below left of=divide, xshift=-1.5cm, yshift=-0.8cm, fill=gray4] {\parbox{4cm}{\centering\textbf{Compute Mean} \\ \parbox{4cm}{\centering \tiny \centering $\mu_k(t)$}}};
    \node (std) [process, below of=divide, yshift=-2cm, fill=gray4] {\parbox{4cm}{\centering \textbf{Compute Standard Deviation} \\ \tiny \centering $\sigma_k(t), \sigma_k(t + n\Delta t)$}};
    \node (corr) [process, below right of=divide, xshift=1.5cm, yshift=-0.8cm, fill=gray4] {\parbox{4cm}{\centering\textbf{Compute Correlations} \\ \tiny \centering $\rho_k(t, n\Delta t), \rho_k(t, n\Delta t - T), \rho_k(t + n\Delta t, T)$}};
    \node (comb) [draw, dashed, thick, rounded corners, fit=(mean) (std) (corr)] {};
    \node (lambda) [process, below of=std, yshift=-0.3cm, fill=gray6, text=white] {\parbox{11cm}{\centering\textbf{Compute $\lambda_k$} \\ \tiny \centering $\frac{(\mu_k(t) - \mu_k(t+n\Delta t))^2 + \sigma_k^2(t+n\Delta t) (1 - \rho_k(t+n\Delta t, T)) + \sigma_k(t) \sigma_k(t+n\Delta t) (\rho_k(t, n\Delta t) - \rho_k(t, n\Delta t-T))}{(\mu_k(t) - \mu_k(t+n\Delta t))^2 + \sigma_k^2(t) + \sigma_k^2(t+n\Delta t) - 2 \rho_k(t, n\Delta t-T) \sigma_k(t) \sigma_k(t+n\Delta t)}$}};
    \node (pcomb) [process, below of=lambda, yshift=-0.1cm, fill=gray6, text=white] {\parbox{5.5cm}{\centering\textbf{Prediction with $P^{\circlearrowleft}_{\text{BLEND}}$} \\ \tiny \centering $\hat{X}_k(t + n\Delta t) = (1 - \lambda_k) P^{\circlearrowleft}(X)(t) + \lambda_k P(X)(t)$}};
    \draw [arrow] (start) -- (pcalc.north);
    \draw [arrow] (start) -- (divide.north);
    \draw [arrow] (pcalc) |- (pcomb);
    \draw [arrow] (divide) -- (comb);
    \draw [arrow] (comb.south) -- (lambda.north);
    \draw [arrow] (lambda) -- (pcomb);
\end{tikzpicture}}
   \caption{Cyclostationary Case ($P^{\circlearrowleft}_{\text{BLEND}}$)}
    \end{subfigure}
    \begin{subfigure}[b]{\textwidth}
        \centering
        \scalebox{0.6}{
\begin{tikzpicture}[node distance=1.8cm]
    \node (start) [process, fill=gray1] {\parbox{6cm}{\textbf{Cyclostationary Time Series} \\ \tiny \centering Period $T$, Time Step $\Delta t$, horizon $n\Delta t$}};
    \node (pcalc) [process, below left of=start, xshift=-2cm, yshift=-3.5cm, fill=gray2] {\parbox{4cm}{\centering\textbf{Persistences} \\ \tiny  \parbox{4cm} {\centering$P(X)(t)=X(t)$  $P^{\circlearrowleft}(X)(t)=X(t - T + n\Delta t)$}}};
    \node (divide) [process, below right of=start, xshift=2cm, yshift=-0.5cm, fill=gray3] {\parbox{4cm}{\centering\textbf{Divide the Series} \\ \tiny \centering into $k = T / \Delta t$ Sub-Series}};
    \node (corr) [process, below right of=divide, xshift=-1.3cm, yshift=-1.7cm, fill=gray4] {\parbox{4cm}{\centering\textbf{Compute Correlations} \\ \tiny \centering $\rho_k(t, n\Delta t)$}};
    \node (comb) [draw, dashed, thick, rounded corners, fit=(corr)] {};
    \node (lambda) [process, below of=std, yshift=-0.3cm, fill=gray6, text=white] {\parbox{6cm}{\centering\textbf{Compute $\lambda_k$} \\ \tiny \centering $ \frac{1}{2}(1 + \rho_k(t, n\Delta t)).$}};
    \node (pcomb) [process, below of=lambda, yshift=-0.1cm, fill=gray6, text=white] {\parbox{5.5cm}{\centering\textbf{Prediction with $\tilde{P}^{\circlearrowleft}_{\text{BLEND}}$} \\ \tiny \centering $\hat{X}_k(t + n\Delta t) = (1 - \lambda_k) P^{\circlearrowleft}(X)(t) + \lambda_k P(X)(t)$}};
    \draw [arrow] (start) -- (pcalc.north);
    \draw [arrow] (start) -- (divide.north);
    \draw [arrow] (pcalc) |- (pcomb);
    \draw [arrow] (divide) -- (comb);
    \draw [arrow] (comb.south) -- (lambda.north);
    \draw [arrow] (lambda) -- (pcomb);
\end{tikzpicture}}
   \caption{Simplified Case ($\tilde{P}^{\circlearrowleft}_{\text{BLEND}}$)}
    \end{subfigure}    
    \caption{Flowcharts illustrating the forecasting process for the \texttt{BLEND} Persistence operator in three cases: (a) stationary, (b) cyclostationary, and (c) simplified cyclostationary. The stationary case assumes constant statistical properties, while the cyclostationary case integrates periodic variations in mean, variance, and correlations. The simplified cyclostationary case uses an empirical approximation to further reduce complexity while retaining predictive efficiency.}
    \label{fig:both_diagrams}
\end{figure*}
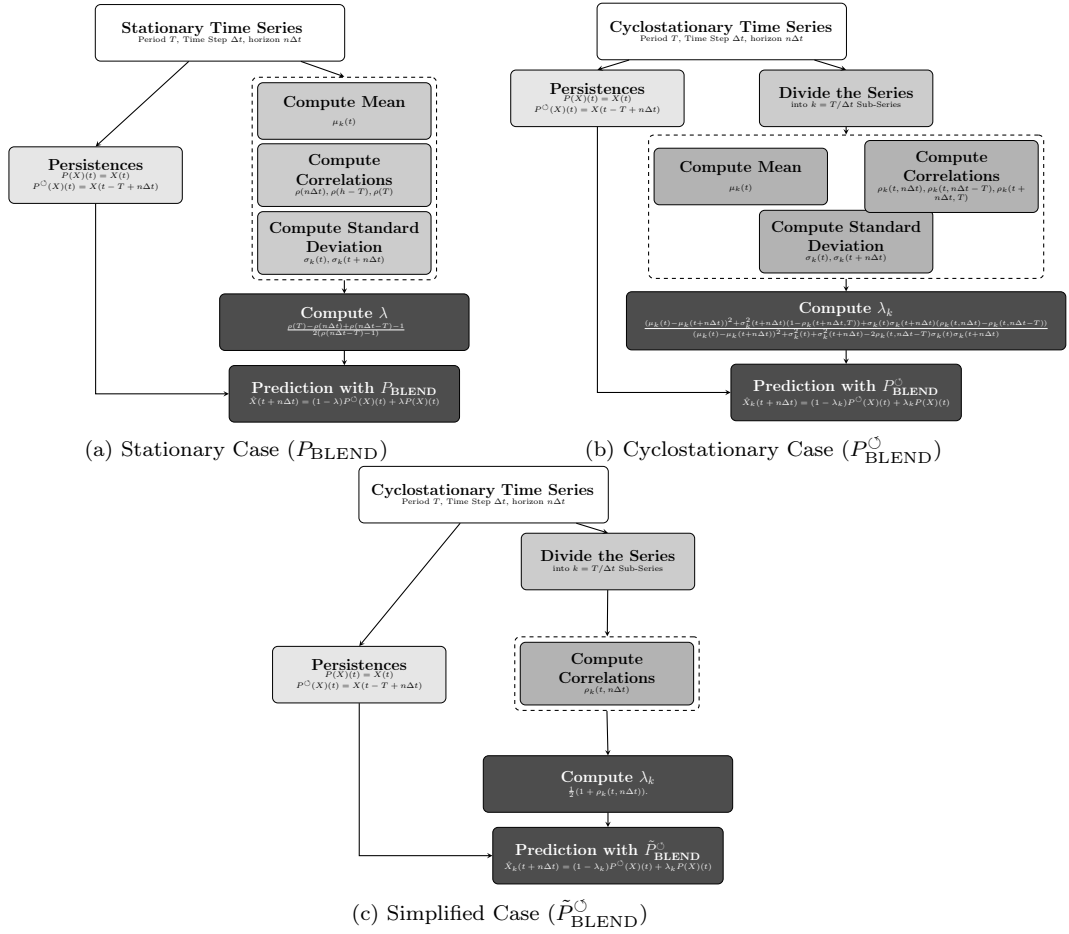
%%%%%%%%%%%%%%%%%%%%%%%%%%%%%%%%%%%%%%%%%%%%%%%%%%%%%%%%%%%%%%%%%%%%%%%%%%%%%%%%%%%%%%%%%%%%%%%%%%%%%%%
%                                                                                                     %
%                                          Results                                                    %
%                                                                                                     %
%%%%%%%%%%%%%%%%%%%%%%%%%%%%%%%%%%%%%%%%%%%%%%%%%%%%%%%%%%%%%%%%%%%%%%%%%%%%%%%%%%%%%%%%%%%%%%%%%%%%%%%
\section{Modeling Results and Analysis}
\label{sec3}
Two complementary datasets are employed to assess the proposed operators across synthetic and measured conditions. The first consists of 100 controlled cyclostationary series generated through a hybrid $\mathtt{bootstrap}$–$\mathtt{Monte\text{-}Carlo}$ scheme, enabling modulation of periodicity, noise, and correlation structure.  The second dataset corresponds to an empirical use case in renewable energy forecasting, solar irradiation, based on 30-min observations from 68 $\mathtt{SIAR}$ meteorological stations in Spain \citep{DESPOTOVIC2024} (Figure~\ref{map}, data available in \url{https://observatorioregadio.gob.es/fr/outils/siar/}). Spatial consistency is ensured by ordinary kriging, following the empirical variogram of each subregion. Forecast horizons span 30 min to 6 h in half hour increments.  An ad hoc quality control following \cite{app12178529} guarantees data integrity. Cyclostationary statistics and learning were estimated using two full years of data. Less than 5\% of the observations were missing or rejected during quality control procedures and were excluded from the analysis. The final year serves as out-of-sample validation under an \emph{incast} configuration, reconstructing past states from subsequent observations to quantify retrospective forecast skill. All methodological steps are fully described through analytical formulations and pseudocode, enabling independent reproduction of the experiments in any programming environment. The $\mathtt{SIAR}$ observational data used for validation are publicly accessible via the official network portal; however, redistribution is subject to the data provider policy. All MATLAB source code implementing the operators described in this section is openly available at \url{https://github.com/cyrilvoyant/cyclostationary-forecasting-matlab} (DOI: \url{https://doi.org/10.5281/zenodo.18812334}), enabling full reproducibility of the presented results.
\begin{figure}[!t]
    \centering
    \includegraphics[width=0.8\textwidth]{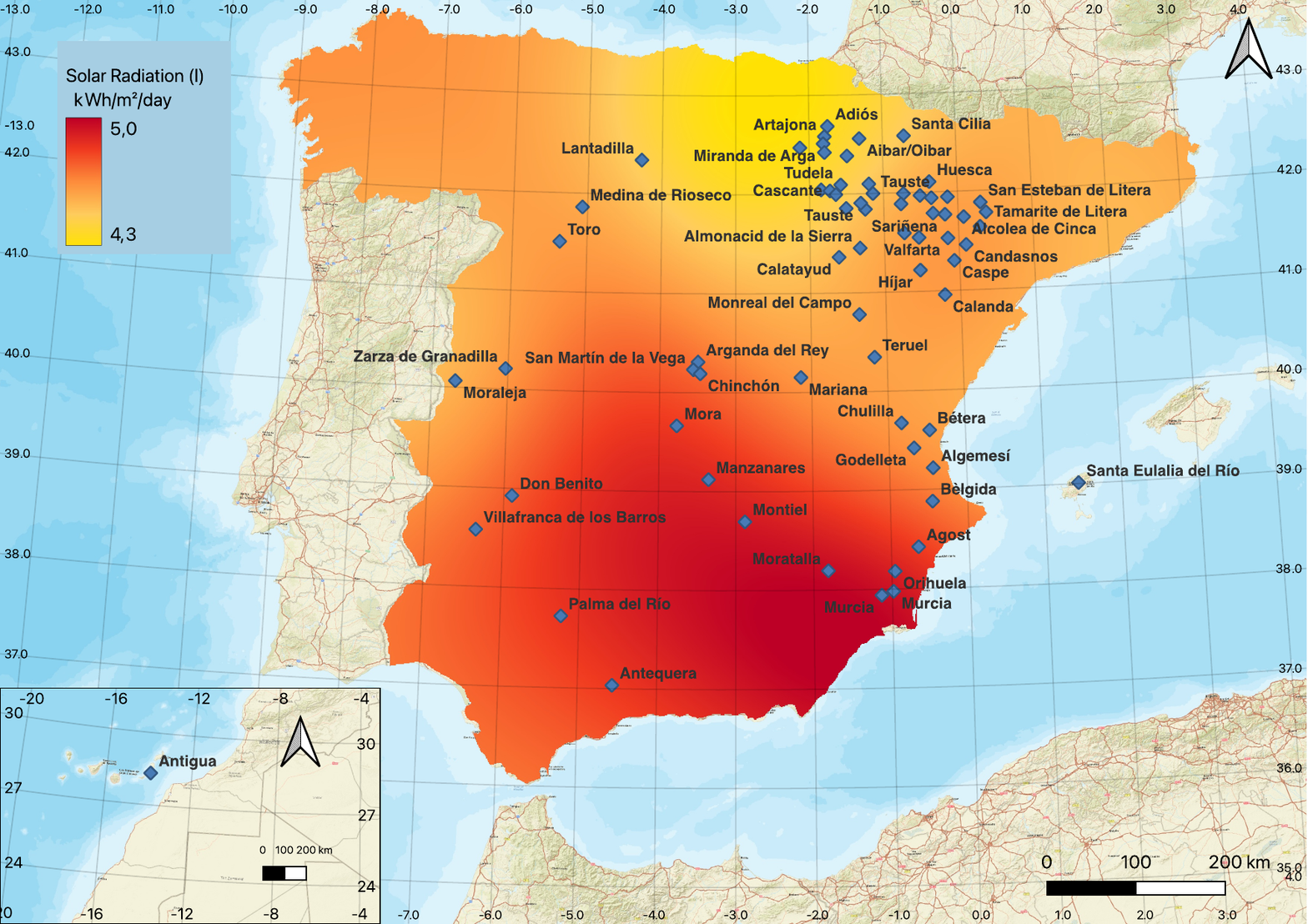} 
    \caption{Geographical distribution of the 68 $\mathtt{SIAR}$ weather stations and kriging surface of daytime averaged $I$ (2017-2020)}
    \label{map}
\end{figure}

\subsection{Generation of Synthetic Cyclostationary Time Series}
\label{sec:synth}
The Monte Carlo analysis is based on 100 independent cyclostationary time series. Additional experiments conducted with 500 replications confirmed the convergence and stability of the results. These additional tests did not alter the conclusions and are therefore not shown for brevity. Algorithm~\ref{alg:synthetic_series} generates synthetic periodically stationary series $I(t)$, $t=1,\ldots,N$, by combining a deterministic periodic component $\tilde{X}_t$ with correlated noise filtered through a low pass smoother. 
\newcommand{\com}[1]{\hfill\textcolor{blue}{\(\triangleright\,\text{#1}\)}}
\begin{doublealgorithm}[\textbf{Synthetic Cyclostationary Time Series Generation and Evaluation}]
\textbf{Input:} \\
\quad $N$: Number of points in the time series (\textit{e.g.} $N = 5000$), \\
\quad $T$: Period of the harmonic component (\textit{e.g.} $T = 40$), \\
\quad $A$: Noise amplitude scaling factor (\textit{e.g.} $A = 0.5$), \\
\quad $L$: Window size for moving average filter (\textit{e.g.} $L = 10$).\\
\textbf{Output:} \\
$X(t)$: Generated synthetic time series. \\
Metrics: $\mathtt{CV}$ (Coefficient of Variation), $\mathtt{MAR}$ (Mean Absolute Return), $\mathtt{RMSE}$ (Root Mean Square Error for periodic trend predictor), ${\rho} (1)$ (Correlation for $n \Delta t =1$).\\
\textbf{Algorithm Steps:}
\begin{algorithmic}[1]
\For{$t = 1, \ldots, N$} 
    \State Compute periodic component: $\tilde{X}(t) \gets 1000\max\left(0, \sin\left(\frac{2\pi t}{T}\right)\right)$
    \com{Ensure non-negativity}
    \State Generate Gaussian noise: $\mathtt{WN}(t) \sim \mathcal{N}(1, 1)$
    \com{White noise centered around 1}
    \State Bound noise: $\mathtt{WN}(t) \gets \min(\max(\mathtt{WN}(t), 0.2), 1.1)$
    \com{Ensure noise is between 0.2 and 1.1}
    \State Low-pass filter: 
    $\mathtt{LPF}(\mathtt{WN}(t)) \gets \frac{1}{L} \sum_{k=0}^{L-1} \mathtt{WN}(t-k), \quad \text{if } t \geq L$
    \com{$L$-size window Smoothing}
    \State Scale the filtered noise: $\mathtt{LPF}(\mathtt{WN}(t)) \gets A \cdot \mathtt{LPF}({WN}(t))$
    \com{Adjust with $A$ scaling factor}
    \State Combine signal and noise: $X(t) \gets \max(0, \tilde{X}(t) \cdot \mathtt{LPF}({WN}(t)))$
    \com{Ensure non-negativity}
\EndFor
\State Compute metrics: 
    \begin{equation*}\mathtt{CV} = \frac{\sigma}{\mu}, \quad   \mathtt{MAR} = \frac{1}{N-1} \sum_{t=2}^N |X(t) - X(t-1)|, \quad    \mathtt{RMSE} = \sqrt{\frac{1}{N} \sum_{t=1}^N (X(t) - \tilde{X}(t))^2}, \quad {\rho} (1) = \frac{C(1)}{{\sigma}^2}.\end{equation*}
\State \textbf{Return:} $X(t)$, $\mathtt{CV}$, $\mathtt{MAR}$, $\mathtt{RMSE}$, ${\rho} (1)$.
\begin{center}
\begin{minipage}[c]{0.4\textwidth}
  \centering
  \includegraphics[width=\linewidth]{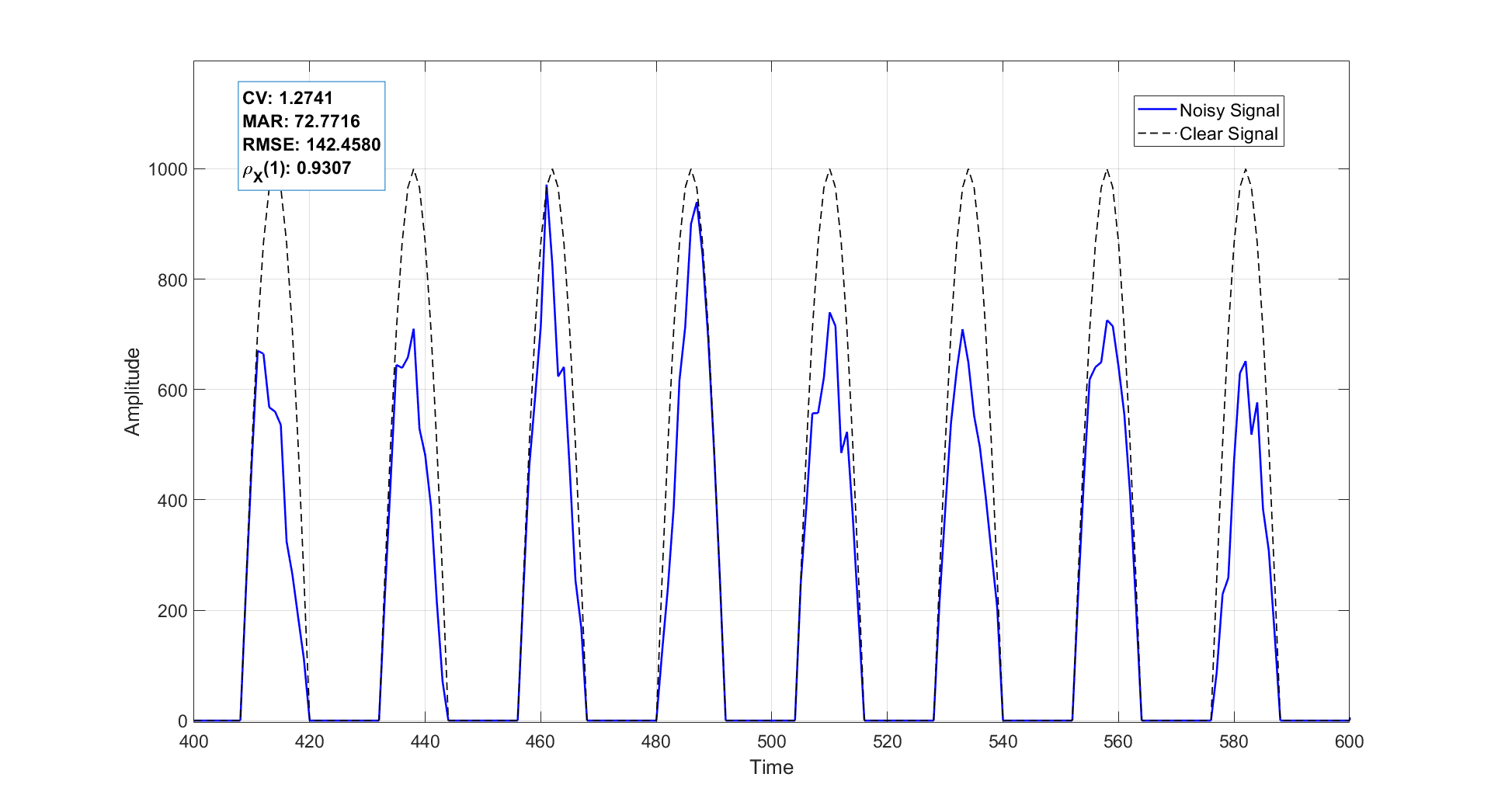}
\end{minipage}%
\begin{minipage}[c]{0.55\textwidth}
  \raggedleft
  \com{\textbf{Parameters:}}\\[2pt]
  \com{$N = 5000$}\\
  \com{$T = 40$}\\
  \com{$A = 0.5$}\\
  \com{$L = 10$}
\end{minipage}
\end{center}
\end{algorithmic}
    \label{alg:synthetic_series}
\end{doublealgorithm}

The resulting signals reproduce periodic patterns with controlled stochastic variability, emulating typical cyclostationary behavior.  Four statistical indicators characterize each realization: the Coefficient of Variation ($\mathtt{CV}$) for relative variability, the Mean Absolute Return ($\mathtt{MAR}$) for temporal volatility, the Root Mean Square Error ($\mathtt{RMSE}$) for deviation from the periodic trend, and the lag-1 autocorrelation $\rho(1)$ for short term dependence.  This design ensures rigorous control of both deterministic and stochastic properties within the generated ensemble. Figure~\ref{cyclic_parameters_plot} illustrates cyclic parameters and forecasting performance for a noisy periodic signal with period $T=24$. The left panel shows the cyclic mean, standard deviation, and correlation coefficient (undefined for all-zero phase bins), each in raw and smoothed form to highlight the periodic structure.  The right panel compares representative forecasting operators: $\mathtt{P}$, $\mathtt{P}^{\circlearrowleft}$, $\mathtt{P_S}$, $\mathtt{P}_{\mathtt{CLIPER}}$, $\mathtt{P}^{\circlearrowleft}_{\mathtt{CLIPER}}$, $\tilde{\mathtt{P}}^{\circlearrowleft}_{\mathtt{CLIPER}}$, $\mathtt{P}_{\mathtt{BLEND}}$, $\mathtt{P}^{\circlearrowleft}_{\mathtt{BLEND}}$, and $\tilde{\mathtt{P}}^{\circlearrowleft}_{\mathtt{BLEND}}$.  These results demonstrate the operators’ ability to capture both short term dependencies and periodic structures.  Cyclic parameters summarize the statistical organization of the process, while the comparative forecasts highlight differences in predictive skill and temporal symmetry preservation.
\begin{figure}[H]
    \centering
         \includegraphics[width=\textwidth]{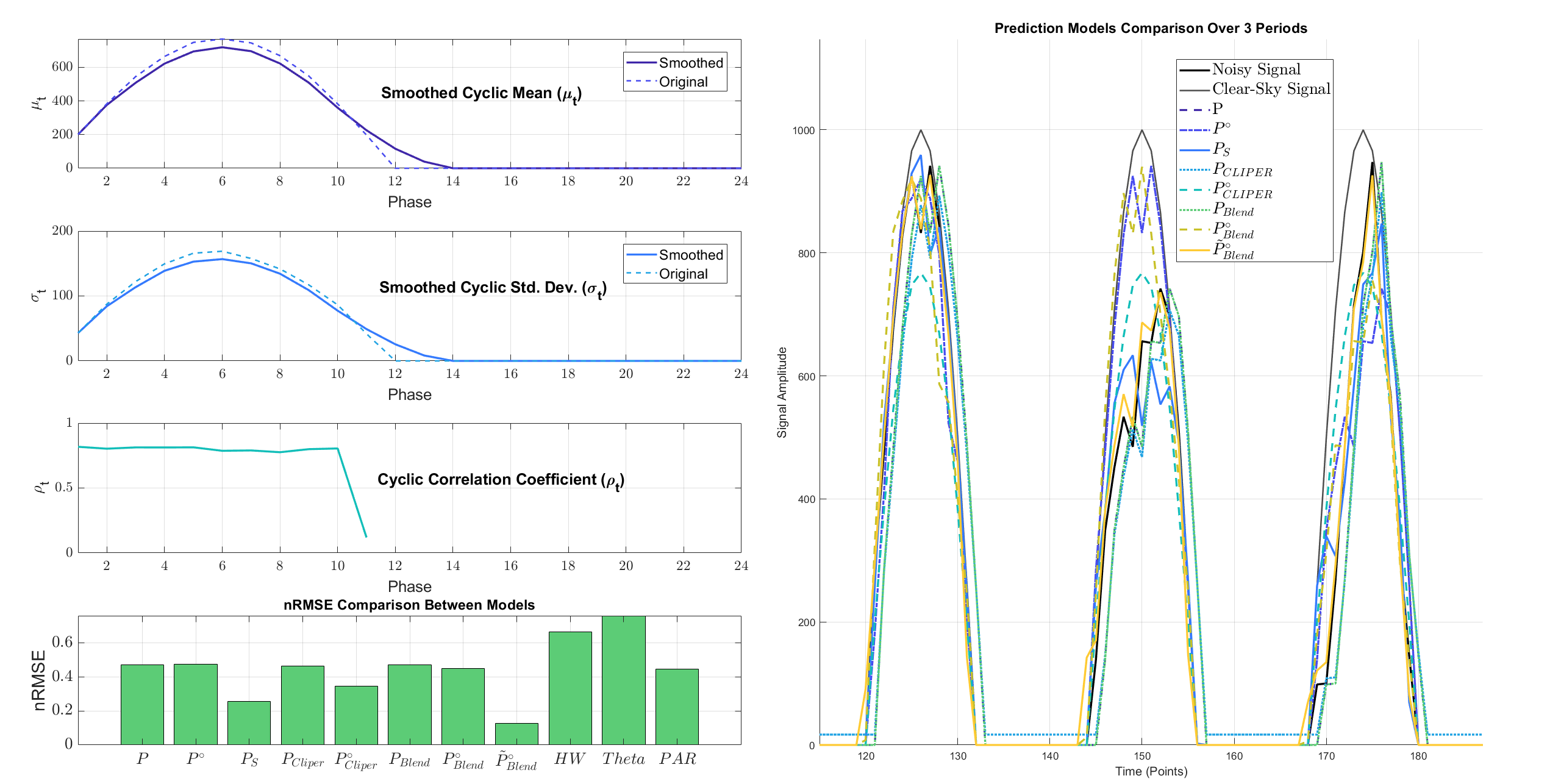}
    \caption{Cyclic parameters (mean $\mu(t)$, standard deviation $\sigma(t)$, and correlation coefficient $\rho(t)$) shown on the left and prediction model comparisons (right) are applied to a noisy periodic signal ($T=24$, $n=1$ and $\Delta t=1$). The left side illustrates the smoothed and original cyclic statistics, while the right side displays the performance of models.}
    \label{cyclic_parameters_plot}
\end{figure}

\subsection{Benchmark of Reference Models with Synthetic Data}
\label{sec:synthetic_benchmark}
Among the forecasting operators in Table~\ref{tab11}, three reference models were retained for benchmarking: Holt–Winters ($\mathtt{HW}$), the Theta method, and the Periodic Autoregressive model ($\mathtt{PAR}$).  These methods are classical analytical frameworks able to represent trends, seasonality, and cyclostationary dependencies while remaining computationally efficient.  In contrast, the persistence based operators proposed in this work are parameter free and serve as analytical baselines for evaluating the potential gains achievable with model driven approaches. For $\mathtt{HW}$ and Theta, parameters were optimized via $\mathtt{MATLAB}$ \texttt{fmincon} Interior Point solver; ten independent runs were averaged to reduce local minima bias.  The $\mathtt{PAR}$ model was implemented with order $p = T$, exploiting the intrinsic periodicity of non seasonally adjusted signals.  A Monte Carlo protocol of 100 realizations was applied across forecast horizons $n\Delta t \in [1,9]$.  Each series used randomized parameters $T \sim \mathcal{U}(10,50)$ and $L \sim \mathcal{U}(1,10)$ (see Algorithm~\ref{alg:synthetic_series}), to ensure statistical diversity and robustness.  Figure~\ref{fig:nrmse_vs_variability} shows the relation between normalized RMSE and signal variability, with trends fitted by third order polynomials. 
\begin{figure}[H]
    \centering
    \includegraphics[width=\textwidth]{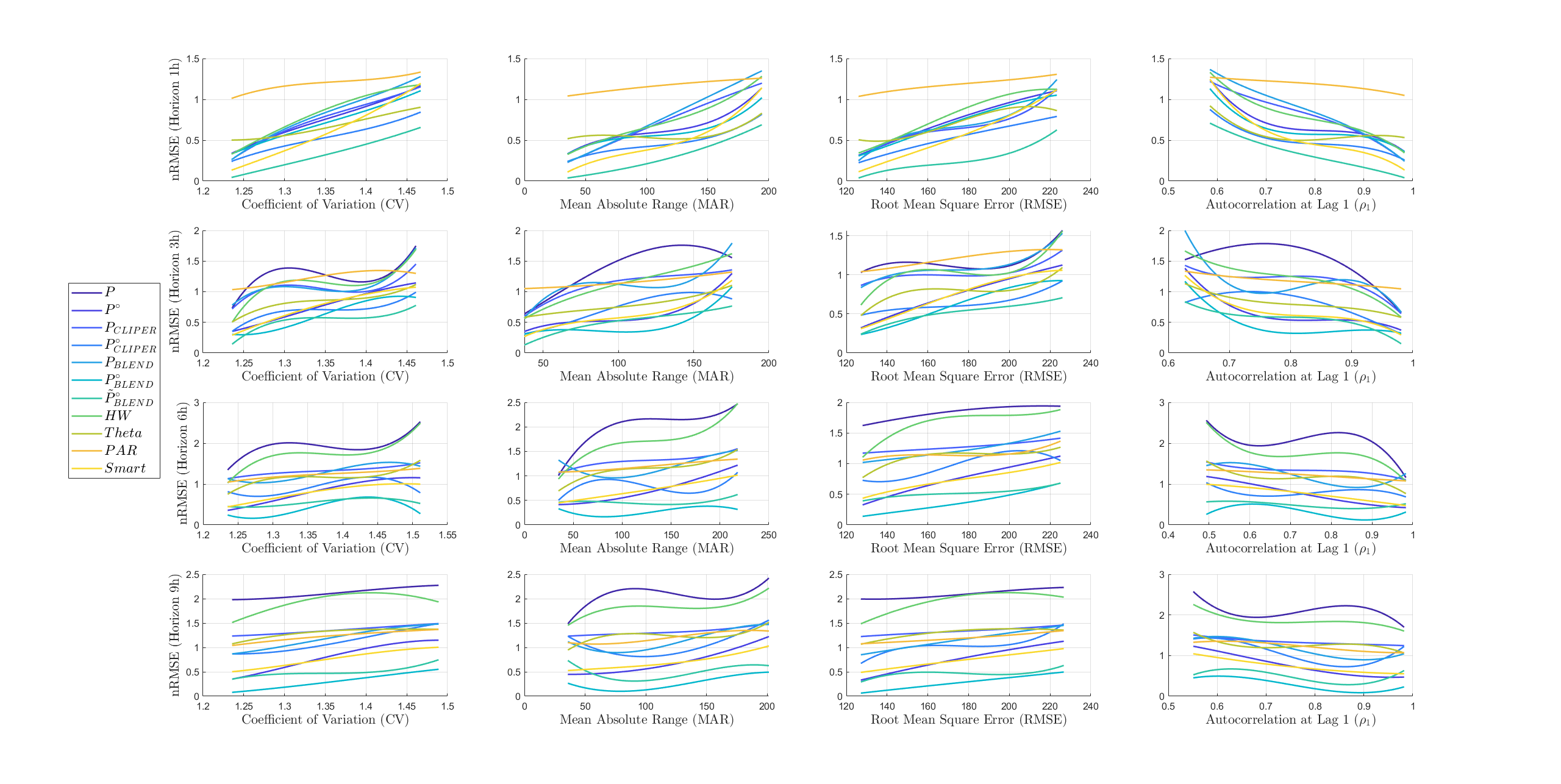}
    \caption{Normalized Root Mean Square Error (\texttt{nRMSE}) plotted against variability metrics for synthetic time series across multiple prediction horizons. 
    Each row represents a specific horizon (1, 3, 6, and 9 hours), and each column corresponds to a variability metric: 
    Coefficient of Variation (\texttt{CV}), Mean Absolute Return (\texttt{MAR}), Root Mean Square Error (\texttt{RMSE}), and Autocorrelation at Lag 1. 
    Trend lines using third order polynomials highlight the relationship between \texttt{nRMSE} and the variability characteristics of the data, 
    illustrating the impact of variability on model performance.}
    \label{fig:nrmse_vs_variability}
\end{figure}
As the horizon increases, $\mathtt{nRMSE}$ sensitivity to variability weakens, reflecting the diminishing relevance of short term correlation beyond $\sim$6 h. Figure~\ref{fig:boxplot_nrmse_synthetic} summarizes the distribution of $\mathtt{nRMSE}$ across models and horizons, confirming the expected performance hierarchy among reference and persistence based operators.
\begin{figure}[H]
    \centering
    \includegraphics[width=\textwidth]{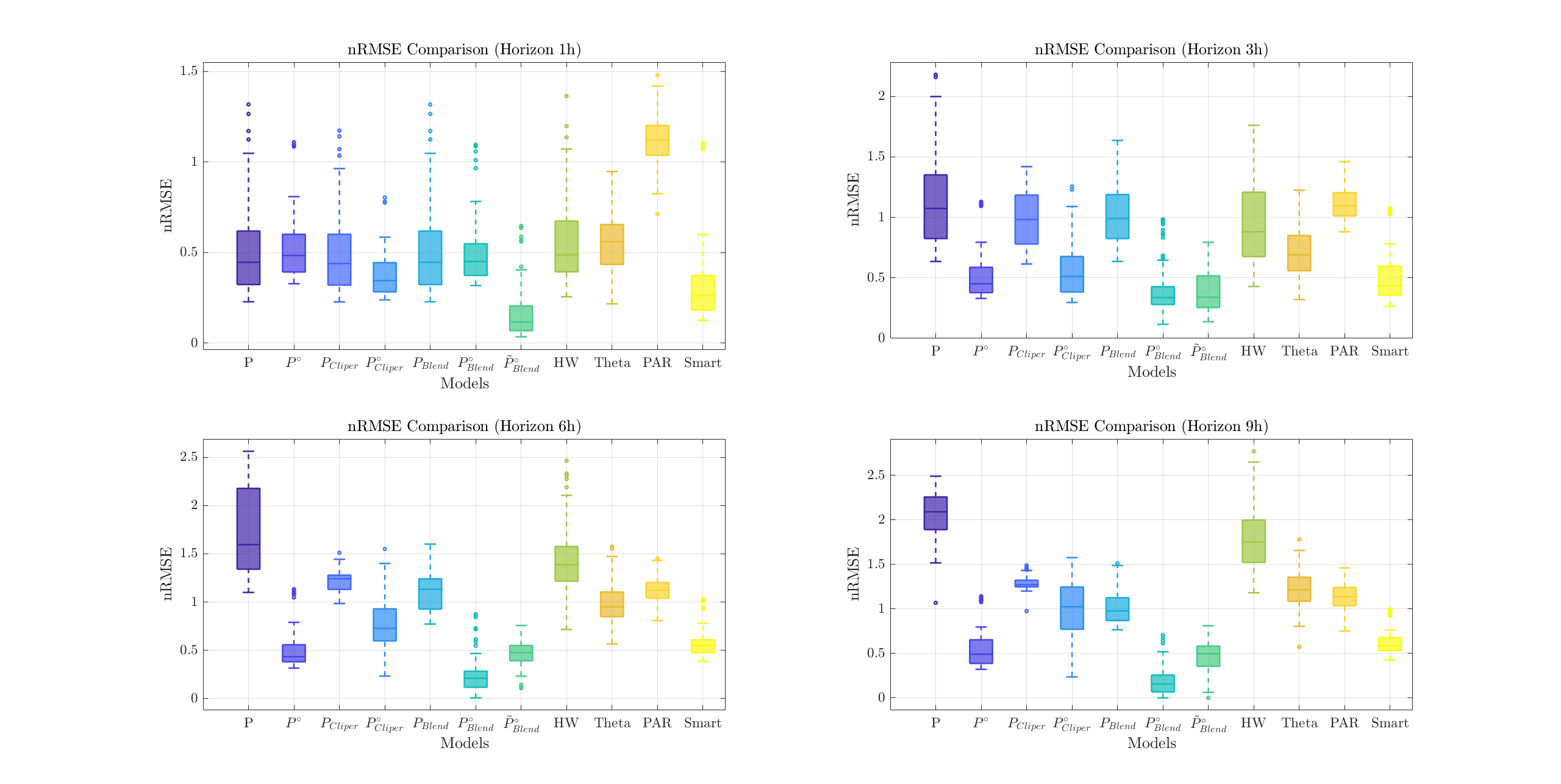}
    \caption{Boxplots of the Normalized Root Mean Square Error (\texttt{nRMSE}) for different models across prediction horizons (1h, 3h, 6h, 9h). The x-axis represents the forecasting models, while the y-axis shows the distribution of \texttt{nRMSE} values across synthetic time series simulations.}
    \label{fig:boxplot_nrmse_synthetic}
\end{figure}
The Kruskal-Wallis test confirms significant differences between models across all horizons (\textit{p}-values $\ll 0.05$). Pairwise Mann-Whitney-Wilcoxon tests further refine these rankings. 
For short term horizons (1h and 3h), $\tilde{\mathtt{P}}^{\circlearrowleft}_{\mathtt{BLEND}}$ and $\mathtt{P}^{\circlearrowleft}_{\mathtt{BLEND}}$ consistently emerge as the best models, with a statistically significant advantage at 1h (\textit{p}-value $\ll 0.05$), but no significant difference at 3h (\textit{p}-value = 0.8613).
For 6h, $\mathtt{P_s}$ and $\mathtt{P}^{\circlearrowleft}$ stand out, with $\mathtt{P_s}$ showing a statistically significant improvement over $\mathtt{P}^{\circlearrowleft}$ (\textit{p}-value $\ll 0.05$).
At 9h, $\mathtt{P}_\mathtt{CLIPER}$ remains the best performing model, significantly outperforming $\mathtt{P}^{\circlearrowleft}$ (\textit{p}-value $\ll 0.05$).
Overall, models can be classified into three categories:
\begin{itemize}
     \item[\(\checkmark\)]  \textbf{Best}: $\tilde{\mathtt{P}}^{\circlearrowleft}_{\mathtt{BLEND}}$, $\mathtt{P}^{\circlearrowleft}_{\mathtt{BLEND}}$, $\mathtt{P_s}$, $\mathtt{P}^{\circlearrowleft}$  
     \,\, \textbf{Med}: $\mathtt{P}^{\circlearrowleft}_\mathtt{CLIPER}$, $\mathtt{Theta}$, $\mathtt{P}^{\mathtt{BLEND}}$  
    \,\,  \textbf{Worst}: $\mathtt{P}$, $\mathtt{P}_\mathtt{CLIPER}$, $\mathtt{PAR}$, $\mathtt{HW}$.  
\end{itemize}
These conclusions remain conditioned by the synthetic experiments, where trend and periodicity are explicitly defined.  Persistence based models demonstrate strong efficiency under such controlled conditions, whereas statistical methods such as $\mathtt{Theta}$, $\mathtt{PAR}$, and $\mathtt{HW}$ show limited adaptability to cyclic trends. Subsequent analyses therefore focus on a restricted subset of operators:
$\tilde{\mathtt{P}}^{\circlearrowleft}_\mathtt{BLEND}$ and $\mathtt{P}^{\circlearrowleft}_\mathtt{BLEND}$, which exhibited the most consistent performance; $\mathtt{P}_{\mathtt{CLIPER}}$, retained for its analytical simplicity;  
and $\mathtt{P_S}$ (Smart Persistence), the most widely used naive predictor in operational solar forecasting. For completeness, the canonical persistence models $\mathtt{P}$ and $\mathtt{P}^{\circlearrowleft}$ are also included, since the latter performs notably well at extended horizons. Two lightweight learning based baselines are incorporated for comparison:  (i) a linear AutoRegressive model ($\mathtt{AR}$, \citep{GUCYETMEZ2025981}) applied to stationary transformed data (\textit{e.g.}, clear sky ratio \citep{solar2040026}) and (ii) an Extreme Learning Machine ($\mathtt{ELM}$, \citep{VOYANT2025113490}) trained directly on raw irradiance without normalization. All other complex or hybrid models are intentionally excluded to maintain focus on physically interpretable and analytically tractable formulations. Finally, convexity and boundedness effects identified in Lemmas~\ref{problem3}--\ref{problem} (\ref{CLIPER_cyclo_statio} \& \ref{BLEND_cyclo_statio}) further justify the preference for the simplified operator $\tilde{\mathtt{P}}^{\circlearrowleft}_{\mathtt{BLEND}}$, which ensures stability and interpretability without compromising analytical coherence.

\subsection{Benchmark of Reference Models with \texttt{SIAR} Data}
This section compares forecasting operators on the $\mathtt{SIAR}$ irradiance dataset (68 stations, 30-minute resolution). In the general formulation, $k_r(t)$ denotes the relative or detrended signal obtained after removal of its deterministic periodic component. In the specific case of solar irradiance, this corresponds to normalization by the clear sky model, \textit{i.e.} $k_r(t) = I(t)/I_{\mathtt{clr}}(t)$, where $I_{\mathtt{clr}}(t)$ represents the theoretical cloud free irradiance. Forecasts are generated every 30 minutes, with horizons ranging from 30 minutes to 6 hours. All models are applied in an incast scheme, without data assimilation or retraining across the test year. The following reference models are retained:
\begin{itemize}
\scriptsize
    \item[$\checkmark$] $\mathtt{P}$ — simple persistence: baseline forecast using last observed irradiance; no adaptation, poor under rapid transitions;
    \item[$\checkmark$] $\mathtt{P}^{\circlearrowleft}$ — cyclic persistence: injects diurnal periodicity without modeling variability; useful for longer horizons;
    \item[$\checkmark$] $\mathtt{P_{BLEND}}$ — naive blend of current and lagged irradiance; simple, robust, and model free, mid level complexity;
    \item[$\checkmark$] $\mathtt{P^{\circ}_{\mathtt{BLEND}}}$ — cyclic version of the blend; better performance under daily trends, high complexity;
    \item[$\checkmark$] $\widetilde{\mathtt{P}}^{\circlearrowleft}_{\mathtt{BLEND}}$ — temporally smoothed cyclostationary blend; balances simplicity and forecast consistency;
    \item[$\checkmark$] $\mathtt{EL}$ — extreme learning machine: data driven model using previous 48~h; best performance without clear sky dependency;
    \item[$\checkmark$] $\mathtt{AR}$ — autoregressive model (order $p$ from $\mathtt{PACF}$); depends on temporal structure and $I_{\mathtt{clr}}$ as regressor;
    \item[$\checkmark$] $\mathtt{Smart}$ — clear sky persistence: corrects for atmospheric attenuation using external $I_{\mathtt{clr}}$ (\textit{e.g.}, McClear);
    \item[$\checkmark$] $\mathtt{P_{CLIPER}}$ — persistence with climatological scaling (convex combination between persistence and mean): uses clear sky index and trend normalization; good average behavior but dependent on sky modeling.
\end{itemize}
Each model outputs forecasts $\widehat{I}(t + n\Delta t)$ from available information at $t$ (or $t + n\Delta t- T$ for periodic inputs). Errors are evaluated against $I(t + n\Delta t)$ using $\mathtt{nRMSE}$ and $\mathtt{nMAE}$ (normalized by the average value of irradiance \citep{YANG202020}).
Trainless models ($\mathtt{P}$, $\mathtt{P}^{\circlearrowleft}$, $\widetilde{\mathtt{P}}^{\circlearrowleft}_{\mathtt{BLEND}}$) offer strong interpretability and zero configuration, suitable for embedded forecasting. In contrast, $\mathtt{AR}$ and $\mathtt{EL}$ use historical information and adapt to local patterns, at the cost of increased complexity. The benchmark thus covers both operational simplicity and adaptive modeling capabilities.

Figure~\ref{fig:nrmse_nmae_horizon} indicates that the machine learning model ($\mathtt{EL}$) delivers the lowest average $\mathtt{nRMSE}$ and $\mathtt{nMAE}$ across all forecast horizons, outperforming all other deterministic models despite not using any clear sky input. 
\begin{figure}[H]
    \centering
    \includegraphics[width=0.8\textwidth]{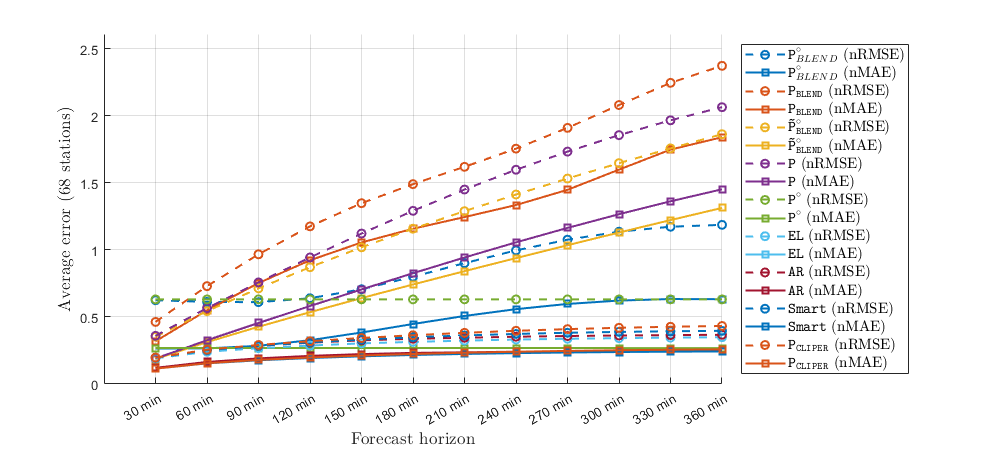}
    \caption{Average normalized forecasting errors ($\mathtt{nRMSE}$ and $\mathtt{nMAE}$) as a function of lead time, ranging from 30 to 360 minutes. Each curve corresponds to a forecasting model, and values are averaged over the full $\mathtt{SIAR}$ network (68 stations). This representation emphasizes the relative scalability and asymptotic behavior of the error dynamics across models.}
    \label{fig:nrmse_nmae_horizon}
\end{figure}
By contrast, the classical reference models (Smart Persistence ($\mathtt{P_S}$), autoregressive ($\mathtt{AR}$), and $\mathtt{P_{CLIPER}}$) each depend on an external clear sky irradiance model (\textit{e.g.}, $\mathtt{McClear}$ \citep{EISSA201517}).  In our implementation, $\mathtt{P_{CLIPER}}$ operates on the stationary clear sky index $k_r(t) = I(t)/I_\mathtt{clr}(t)$, ensuring an idealized trend baseline.  
This dependence improves performance when clear sky inputs are accurate but may introduce bias or degradation if unavailable or inconsistent, since the $\mathtt{BLEND}$ family operates directly on raw irradiance without external normalization. In persistence based baselines, introducing cyclostationary adjustments yields systematic gains: the phase corrected variants ($\mathtt{P}^{\circlearrowleft}$, $\mathtt{P^{\circlearrowleft}_{CLIPER}}$) consistently outperform their naive counterparts ($\mathtt{P}$, $\mathtt{P_{CLIPER}}$), with widening performance gaps at extended horizons.  This is confirmed in Figure~\ref{fig:violinplots}, where cyclo adjusted models display narrower $\mathtt{nRMSE}$ distributions at longer lead times.  Among these, the simplified $\tilde{\mathtt{P}}^{\circlearrowleft}_{\mathtt{BLEND}}$ offers the most practical balance, maintaining accuracy across horizons without clear sky inputs and with minimal computational cost.  Non-parametric Wilcoxon tests ($p<0.01$) confirm the statistical significance of inter model differences, underscoring that integrating periodic symmetry into persistence baselines yields measurable and operationally meaningful improvements.
\begin{figure}[H]
    \centering
    \includegraphics[width=0.7\textwidth]{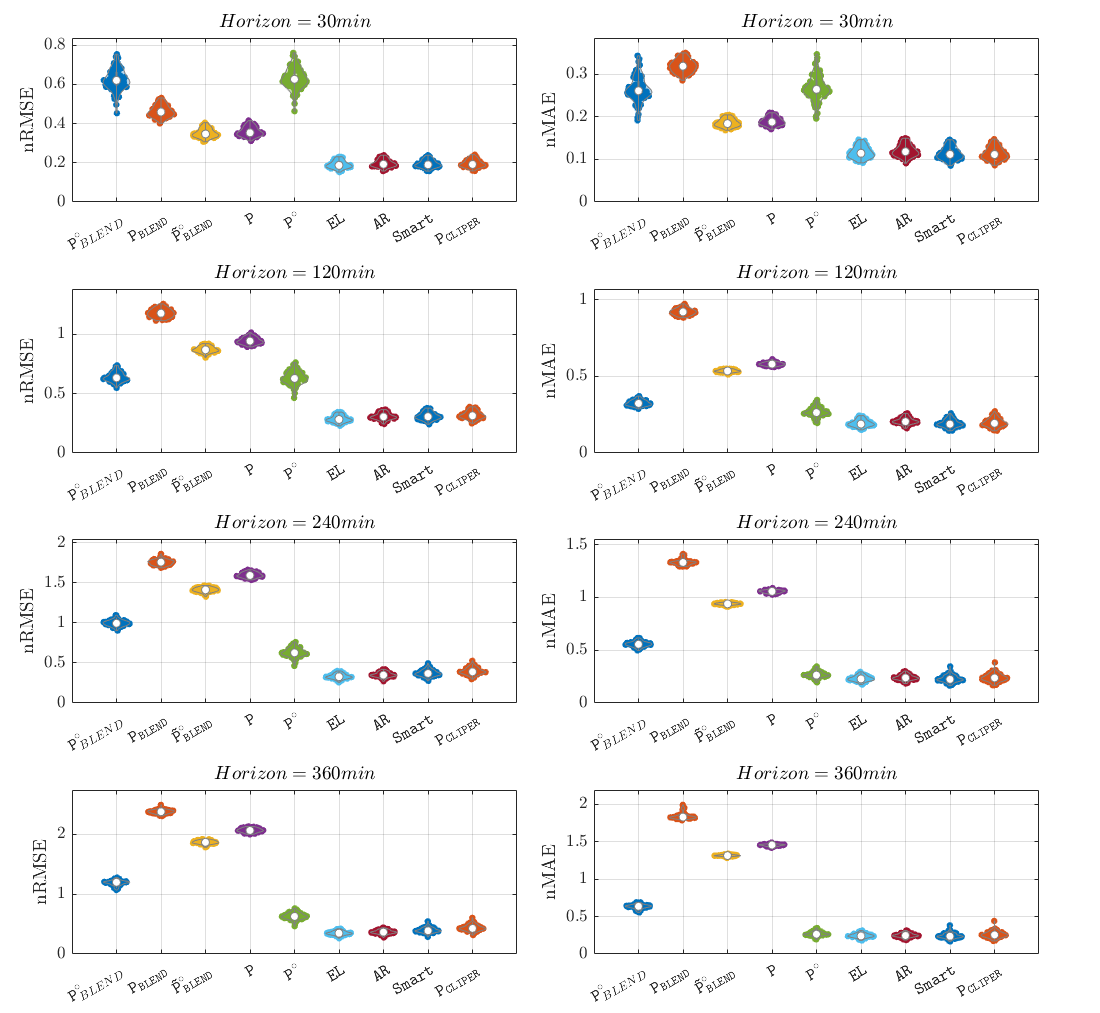}
    \caption{Distribution of normalized forecasting errors across stations for four selected horizons (30, 120, 240, and 360 minutes). Violin plots summarize both the central tendency and dispersion of the $\mathtt{nRMSE}$ (left column) and $\mathtt{nMAE}$ (right column) for each model. This multidimensional view reveals the sensitivity of models to temporal scales and spatial variability.}
    \label{fig:violinplots}
\end{figure}
%%%%%%%%%%%%%%%%%%%%%%%%%%%%%%%%%%%%%%%%%%%%%%%%%%%%%%%%%%%%%%%%%%%%%%%%%%%%%%%%%%%%%%%%%%%%%%%%%%%%%%%
%                                                                                                     %
%                                          Conclusion                                                 %
%                                                                                                     %
%%%%%%%%%%%%%%%%%%%%%%%%%%%%%%%%%%%%%%%%%%%%%%%%%%%%%%%%%%%%%%%%%%%%%%%%%%%%%%%%%%%%%%%%%%%%%%%%%%%%%%%
\section{Conclusions}
\label{sec4}
This work establishes a symmetry based analytical framework for forecasting cyclostationary processes in energy systems. By deriving the $\mathtt{BLEND}$ operator family (together with its simplified variant $\widetilde{\mathtt{P}}^{\circlearrowleft}_{\mathtt{BLEND}}$) the transition between stationary and periodically stationary regimes is exposed through a single closed form parameter, $\tilde{\lambda}^{\circlearrowleft}_{\mathtt{BLEND}}=\tfrac{1}{2}(1+\rho)$. This coefficient expresses the intrinsic symmetry between instantaneous memory and phase aligned recurrence, defining the analytical limit of training free forecasting under mean square optimality. The resulting operators preserve periodic covariance, remain physically interpretable, and operate without any external or empirical calibration.
Validation on synthetic and measured solar irradiation datasets demonstrates that explicitly embedding cyclic variance and covariance terms yields measurable accuracy gains over classical persistence baselines, even in the absence of trend normalization (ratio to trend, \textit{e.g.} clear sky model in solar energy prediction). While reference approaches such as $\mathtt{AR(p)}$, $\mathtt{EL}$, or $\mathtt{P_S}$ can outperform under optimized conditions, they require parameter fitting, prior normalization, or learned representations that limit reproducibility. In contrast, the $\mathtt{BLEND}$ family achieves comparable robustness at negligible computational cost, providing a reproducible analytical baseline for deterministic forecasting across horizons. Beyond the specific case of solar irradiance, this framework applies to any energy or environmental process driven by cyclic forcing (such as diurnal demand, wind power modulation, or temperature driven loads) where preserving phase dependent statistics is important. It thus provides a unifying mathematical kernel for symmetry aware prediction, bridging physics based and data driven paradigms. Future work will address time varying periodicity, multivariate coupling, and ensemble extensions, paving the way toward fully interpretable, symmetry consistent forecasting architectures across renewable and hybrid energy systems.

\section*{Author contributions}
\begin{itemize}
    \item \textbf{Cyril Voyant:} Conceptualization, Project leader, Methodology, Software, Formal analysis, Investigation, Writing – original draft.
    \item \textbf{Luis Antonio Garcia-Gutierrez:} Writing – review \& editing, Data curation, Software, Validation, Writing – original draft.
    \item \textbf{Candice Banes:} Conceptualization, Writing – review \& editing, Methodology, Formal analysis, Data curation. 
    \item \textbf{Gilles Notton:} Conceptualization, Validation, Writing – review \& editing, Formal analysis.
    \item \textbf{Milan Despotovic:} Conceptualization, Validation, Writing – review \& editing, Data curation.
    \item \textbf{Zaher Mundher Yaseen:} Conceptualization, Validation, Supervision, Writing – review \& editing, Formal analysis.
\end{itemize}

\section*{Competing interests} 
The authors declare that they have no known competing financial interests or personal relationships that could have appeared to influence the work reported in this paper. 

\section*{Data availability}
The datasets generated during and/or analysed during the current study are available from the corresponding author on reasonable request. 

\iffalse
\section*{Acknowledgements}
This research was partially funded by the $\mathtt{ANR}$ under the $\mathtt{FINE4CAST}$ project (ANR reference: $\mathtt{22\text{-}PETA\text{-}0008}$), whose support is gratefully acknowledged. We also thank \texttt{Elsevier} and the \texttt{Scopus} platform for providing structured access to bibliographic metadata, which was essential to the successful implementation of the text mining pipeline in this study.
\fi
%%%%%%%%%%%%%%%%%%%%%%%%%%%%%%%%%%%%%%%%%%%%%%%%%%%%%%%%%%%%%%%%%%%%%%%%%%%%%%%%%%%%%%%%%%%%%%%%%%%%%%%
%                                                                                                     %
%                                          APPENDICES                                                 %
%                                                                                                     %
%%%%%%%%%%%%%%%%%%%%%%%%%%%%%%%%%%%%%%%%%%%%%%%%%%%%%%%%%%%%%%%%%%%%%%%%%%%%%%%%%%%%%%%%%%%%%%%%%%%%%%%
\appendix
\section{Complexity of an Algorithm}
\label{complexity}
The complexity of an algorithm measures the time or memory space required as a function of the input size $n$. The $\mathcal{O}(.)$ notation is used to express this complexity by estimating how quickly the execution time grows as $n$ increases. Considering an algorithm with $\mathcal{O}(n)$, its execution time increases proportionally to the input size.
For example, in an Autoregressive ($\mathtt{AR}$) model of order $p$, the next value $I(t+1)$ is forecasted based on a linear combination of the previous $p$ values: \begin{equation}
I(t+1) = \phi_1 I(t) + \phi_2 I(t-1) + \dots + \phi_p I(t-p+1) + \epsilon(t),
\end{equation}
where $\phi_i$ are the coefficients and $\epsilon(t)$ is a white noise term. The training complexity (\cyr{$n \gg p$}) of the $\mathtt{AR}$ model is $\mathcal{O}(n\mathtt{P}^2)$ (inversion of $\mathbf{X}^\top \mathbf{X}$ in the least squares optimization process), as fitting the model involves solving a system of equations derived from $n$ observations, requiring the inversion of a $p \times p$ matrix. The forecasting complexity is $\mathcal{O}(p)$, as it depends only on the $p$ most recent observations.

\cyr{
\subsection{Recursive state-space models (\texttt{HW}, \texttt{ETS}, Theta)}
These models rely on recursive updates of a small number of state variables (level $L(t)$, trend $b(t)$, seasonal component $S(t)$). Each update involves a constant number of operations.
Training requires a single pass through the data \(\mathcal{C}_T = \mathcal{O}(n).\) Forecasting over $n$ steps is performed via recursive propagation, \(\mathcal{C}_F = \mathcal{O}(n). \) For the Theta method, $k$ transformed series are processed, so \(\mathcal{C}_T = \mathcal{O}(n k) \text{and} \, \mathcal{C}_F = \mathcal{O}(n).\)}

\cyr{
\subsection{Periodic autoregressive models (\texttt{PAR})}
The $\mathtt{PAR}$ model involves phase-dependent coefficients estimated separately for each phase. This leads to a training complexity dominated by $\mathcal{O}(n p^2)$ when $n/T \gg p$ (see $\mathtt{AR}$ exemple at the begining of the section). Forecasting requires a linear combination of the $p$ most recent values $\mathcal{C}_F = \mathcal{O}(p).$
}

\cyr{
\subsection{Regression-based models (Prophet)}
Prophet is based on additive regression combining trend and seasonal components, typically represented using $k$ Fourier terms. Training involves the construction of a design matrix of size $n \times k$ and parameter estimation through optimization procedures, leading to $\mathcal{C}_T = \mathcal{O}(n k).$ Forecasting consists in evaluating the fitted model over the prediction horizon. Each prediction requires evaluating the trend and seasonal components, resulting in $\mathcal{C}_F = \mathcal{O}(n k),$ which reduces to $\mathcal{O}(n)$ when $k$ is fixed.
}

\cyr{
\subsection{State-space models with seasonal differencing (\texttt{SARIMA})}
$\mathtt{SARIMA}$ models involve autoregressive and moving average components with seasonal differencing. Training relies on parameter estimation procedures combining iterative optimization and matrix operations. The dominant computational cost arises from processing $n$ observations with $p$-dimensional structures, leading to $\mathcal{C}_T = \mathcal{O}(n p^2),$ ($n \gg p$). Forecasting is driven by the autoregressive structure and requires evaluating $p$ past terms $\mathcal{C}_F = \mathcal{O}(p).$
}

\cyr{
\subsection{Models with implementation-dependent complexity (\texttt{STL}, \texttt{TBATS}, \texttt{BEATS})}
For $\mathtt{STL}$, $\mathtt{TBATS}$, and $\mathtt{N\text{-}BEATS}$, both training and forecasting complexities depend strongly on implementation choices (\textit{e.g.}, smoothing parameters, optimization procedures, or network architecture). As a result, no fixed computational complexity is reported for these models.}

\section{Derivation for the \texttt{CLIPER} Operator under Cyclostationary and Stationary Assumptions}
\label{CLIPER_cyclo_statio}
This appendix derives the analytical expression of the optimal blending coefficient under a cyclostationary covariance structure and provides the justification for the simplified formulation adopted in the main text. The $\mathtt{CLIPER} $ operator for a periodically stationary process $I(t)$, where statistical properties exhibit periodicity with a fixed period $T$, predicts future values as:
\begin{equation}
\mathtt{P}^{\circlearrowleft}_{\mathtt{CLIPER}}(I)(t) = \lambda^{\circlearrowleft}_{\mathtt{CLIPER}} \mu(t) + (1 - \lambda^{\circlearrowleft}_{\mathtt{CLIPER}}) I(t),
\end{equation}
where the weighting parameter $\lambda^{\circlearrowleft}_{\mathtt{CLIPER}}$ is derived by minimizing the Mean Squared Error ($\mathtt{MSE}$).
\begin{proof}
The prediction error is defined as:
\begin{equation}
\epsilon(t + n\Delta t) = I(t + n\Delta t) - \mathtt{P}^{\circlearrowleft}_{\mathtt{CLIPER}}(I)(t),
\end{equation}
and the corresponding $\mathtt{MSE} $  is based on $\mathtt{MSE} = \mathbb{E}\left[\epsilon ^2\right]$:
\begin{align*}
&\mathtt{MSE} = \\ &\mathbb{E}\left[\left(I(t + n\Delta t) - \lambda^{\circlearrowleft}_{\mathtt{CLIPER}} \mu(t + n\Delta t) - (1 - \lambda^{\circlearrowleft}_{\mathtt{CLIPER}}) I(t)\right)^2\right].
\end{align*}
Expanding the square, substituting the statistical properties of cyclostationary processes:
\begin{align*}
\mathbb{E}[I(t + n\Delta t)^2] &= \sigma(t + n\Delta t)^2 + \mu(t + n\Delta t)^2, \\
\mathbb{E}[I(t + n\Delta t) I(t)] &= C(t, n\Delta t) + \mu(t + n\Delta t) \mu(t), \\
\mathbb{E}[I(t)^2] &= \sigma(t)^2 + \mu(t)^2, \\
\mathbb{E}[I(t)] &= \mu(t)
\end{align*}
the $\mathtt{MSE} $  becomes:
\begin{align*}
\mathtt{MSE} &= \sigma(t + n\Delta t)^2 + \mu(t + n\Delta t)^2 - 2\lambda^{\circlearrowleft}_{\mathtt{CLIPER}} \mu(t + n\Delta t)^2 \\
&\quad - 2(1 - \lambda^{\circlearrowleft}_{\mathtt{CLIPER}})(C(t, n\Delta t) + \mu(t + n\Delta t) \mu(t)) + \\ &(\lambda^{\circlearrowleft}_{\mathtt{CLIPER}})^2 \mu(t + n\Delta t)^2 + (1 - \lambda^{\circlearrowleft}_{\mathtt{CLIPER}})^2 (\sigma(t)^2 + \mu(t)^2) \\ &\quad + 2\lambda^{\circlearrowleft}_{\mathtt{CLIPER}} (1 - \lambda^{\circlearrowleft}_{\mathtt{CLIPER}}) \mu(t + n\Delta t) \mu(t).
\end{align*}
Taking the derivative of the $\mathtt{MSE} $  with respect to $\lambda^{\circlearrowleft}_{\mathtt{CLIPER}}$ and solving for the minimizer yields:
\begin{align}
&\lambda^{\circlearrowleft}_{\mathtt{CLIPER}} =\\ &\frac{\left(\mu(t + n\Delta t) - \mu(t) \right)^2 + \sigma(t)^2 - \rho (t,n\Delta t) \sigma(t) \sigma (t+n\Delta t)}{\left(\mu(t + n\Delta t) - \mu(t) \right)^2 + \sigma(t)^2 }.
\end{align}
$\qed$
\end{proof}
\begin{lemma}[Convexity of the $\mathtt{MSE} $ ]
The Mean Squared Error ($\mathtt{MSE}$) is convex with respect to $\lambda$ for any cyclostationary or stationary process with finite variance and mean.
\end{lemma}
\begin{proof}
The $\mathtt{MSE} $  can be expressed in the form 
$ \mathtt{MSE}(\lambda^{\circlearrowleft}_{\mathtt{CLIPER}}) = A(\lambda^{\circlearrowleft}_{\mathtt{CLIPER}})^2 + \mathcal{O}(\lambda^{\circlearrowleft}_{\mathtt{CLIPER}}), $
where the leading coefficient is 
$ A = \left(\mu(t + n\Delta t) - \mu(t) \right)^2 + \sigma(t)^2 . $
Since $A$ is a sum of squared terms and variances, it satisfies $A \geq 0$. By the properties of quadratic functions and the Second Derivative Test, the $\mathtt{MSE} $  is convex.
To find the minimum of $ \mathtt{MSE}(\lambda^{\circlearrowleft}_{\mathtt{CLIPER}}) $, we compute its derivative 
$ \frac{d}{d\lambda^{\circlearrowleft}_{\mathtt{CLIPER}}} \mathtt{MSE}(\lambda^{\circlearrowleft}_{\mathtt{CLIPER}}) = 2A\lambda^{\circlearrowleft}_{\mathtt{CLIPER}} - 2B, $
where $B$ is defined from the linear terms in $\lambda^{\circlearrowleft}_{\mathtt{CLIPER}}$. Setting the derivative to zero yields
$ \lambda^{\circlearrowleft}_{\mathtt{CLIPER}} = \frac{B}{A}. $
Since $A$ is composed of squared terms, then $A > 0$ (except in the degenerate case where both the mean difference and variance are zero) and the solution $ \lambda^{\circlearrowleft}_{\mathtt{CLIPER}} = \frac{B}{A} $ uniquely minimizes the $\mathtt{MSE} $ .
\qed
\end{proof}

\begin{lemma}[Boundedness of $\lambda^{\circlearrowleft}_{\mathtt{CLIPER}}$]
\label{problem3}
For a stable cyclostationary process with finite variance and mean, and under the assumption that the variances are not excessively large compared to the means and that the correlation $\rho(t, n\Delta t)$ is not too strong, the parameter $\lambda^{\circlearrowleft}_{\mathtt{CLIPER}}$ remains approximately bounded between $0$ and $1$.
\end{lemma}
\begin{proof}
The denominator is strictly positive, ensuring that $\lambda^{\circlearrowleft}_{\mathtt{CLIPER}}$ is well-defined for non-constant processes.  
The numerator is generally non-negative for stable processes, as the positive terms $\mu(t)^2$, $\mu(t+n\Delta t)^2$, and $\sigma(t)^2$ dominate unless $\rho(t, n\Delta t)$ is close to $1$ and the variances are excessively large compared to the means.  
Similarly, $\lambda^{\circlearrowleft}_{\mathtt{CLIPER}} \leq 1$ holds when the correlation is not too negative and variances are balanced.  
Thus, for most stable cyclostationary processes, $\lambda^{\circlearrowleft}_{\mathtt{CLIPER}} \in [0,1]$, but extreme cases with high variance or strong correlation may violate this bound.
\end{proof}

\begin{corollary}[Stationary Case]
Under stationary assumptions:
\begin{align*}
\mu(t) =& \mu(t + n\Delta t) = \mu, \\ \sigma(t) =& \sigma(t + n\Delta t) = \sigma, \\ \rho(t, n\Delta t) =& \rho(n\Delta t),
\end{align*}
the formula for $\lambda^{\mathtt{CLIPER}}$ simplifies to $\lambda^{\mathtt{CLIPER}} = 1-\rho(n\Delta t).$
\end{corollary}
This demonstrates that the cyclostationary formulation generalizes the stationary case.  
In the stationary case, $\lambda^{\mathtt{CLIPER}}$ depends only on the lagged correlation $\rho(n\Delta t)$, showing that it is purely a function of the temporal dependence structure of the process.

\section{Derivation for the \texttt{BLEND} Persistence Operator under Cyclostationary and Stationary Assumptions}
\label{BLEND_cyclo_statio}
The $\mathtt{BLEND}$ Persistence operator $\mathtt{P}_{\mathtt{BLEND}}$ predicts future values as a convex combination of the cyclic Persistence operator $\mathtt{P}^{\circlearrowleft}$ and the simple Persistence operator $\mathtt{P}$. For cyclostationary processes, it is expressed as:
\begin{equation}
\mathtt{P}^{\circlearrowleft}_{\mathtt{BLEND}}(I)(t) = (1 - \lambda^{\circlearrowleft}_{\mathtt{BLEND}}) \mathtt{P}^{\circlearrowleft}(I)(t) + \lambda^{\circlearrowleft}_{\mathtt{BLEND}} \mathtt{P}(I)(t),
\end{equation}
where $\lambda^{\circlearrowleft}_{\mathtt{BLEND}}$ is determined by minimizing the Mean Squared Error ($\mathtt{MSE}$) for a given forecast horizon $n\Delta t$.
\begin{proof}
The prediction error is given by:
\begin{equation}
\epsilon(t + n\Delta t) = I(t + n\Delta t) - \mathtt{P}^{\circlearrowleft}_{\mathtt{BLEND}}(I)(t),
\end{equation}
which allows us to define the $\mathtt{MSE} $  from:
\begin{align*}
\mathtt{MSE} =  \mathbb{E}\left[\left(I(t + n\Delta t) - (1 - \lambda^{\circlearrowleft}_{\mathtt{BLEND}}) I(t - T + n\Delta t) - \lambda^{\circlearrowleft}_{\mathtt{BLEND}} I(t)\right)^2\right].
\end{align*}
Expanding the square:
\begin{equation}
\begin{aligned}
&\mathtt{MSE}(\lambda^{\circlearrowleft}_{\mathtt{BLEND}})= \\ &\ \mathbb{E}[I(t + n\Delta t)^2] - 2(1 - \lambda^{\circlearrowleft}_{\mathtt{BLEND}})\mathbb{E}[I(t + n\Delta t) I(t - T + n\Delta t)] \\
& - 2\lambda^{\circlearrowleft}_{\mathtt{BLEND}} \mathbb{E}[I(t + n\Delta t) I(t)] + (1 - \lambda^{\circlearrowleft}_{\mathtt{BLEND}})^2 \mathbb{E}[I(t - T + n\Delta t)^2] \\
& + (\lambda^{\circlearrowleft}_{\mathtt{BLEND}})^2 \mathbb{E}[I(t)^2] + 2\lambda^{\circlearrowleft}_{\mathtt{BLEND}}(1 - \lambda^{\circlearrowleft}_{\mathtt{BLEND}})\mathbb{E}[I(t) I(t - T + n\Delta t)].
\end{aligned}
\end{equation}
Introducing key statistical definitions for covariances, correlations, and variance:

\begin{equation}
\begin{aligned}
&C(t + n\Delta t, T) =\\ &\rho(t + n\Delta t, T) \sigma(t + n\Delta t) \sigma(t + n\Delta t + T), \\
&= C(t + n\Delta t - T , T), \\
&= \rho(t + n\Delta t - T, T) \sigma(t + n\Delta t - T) \sigma(t + n\Delta t), \\
&C(t, n\Delta t) = \rho(t, n\Delta t) \sigma(t) \sigma(t + n\Delta t), \\
&C(t - T + n\Delta t, 0) = \sigma^2(t + n\Delta t - T), \\
&C(t, n\Delta t - T) = \rho(t, n\Delta t - T) \sigma(t) \sigma(t + n\Delta t - T), \\
&C(t, 0) = \sigma^2(t), \\
&C(t + n\Delta t, 0) = \sigma^2(t + n\Delta t),
\end{aligned}
\end{equation}
the properties of second-order cyclostationary processes can be expressed as:
\begin{equation}
\begin{aligned}
\mathbb{E}[I(t + n\Delta t) I(t - T + n\Delta t)] &= \rho(t + n\Delta t, T) \sigma(t + n\Delta t) \sigma(t + n\Delta t + T) + \mu(t + n\Delta t) \mu(t - T + n\Delta t), \\
\mathbb{E}[I(t + n\Delta t) I(t)] &= \rho(t, n\Delta t) \sigma(t) \sigma(t + n\Delta t) + \mu(t + n\Delta t) \mu(t), \\
\mathbb{E}[I(t - T + n\Delta t)^2] &= \mu^2(t - T + n\Delta t) + \sigma^2(t + n\Delta t - T), \\
\mathbb{E}[I(t - T + n\Delta t) I(t)] &=  \rho(t, n\Delta t - T) \sigma(t) \sigma(t + n\Delta t - T) + \mu(t - T + n\Delta t) \mu(t), \\
\mathbb{E}[I(t)^2] &= \mu^2(t) + \sigma^2(t),\\
\mathbb{E}[I(t +n\Delta t)^2] &= \mu^2(t + n\Delta t) + \sigma^2(t + n\Delta t).
\end{aligned}
\end{equation}
The optimal $\lambda^{\circlearrowleft}_{\mathtt{BLEND}}$ is obtained by minimizing the $\mathtt{MSE}$:
\begin{equation}
\lambda^{\circlearrowleft}_{\mathtt{BLEND}} =
\frac{
\begin{aligned}
(\mu(t) - \mu(t+n\Delta t))^2 + \sigma^2(t+n\Delta t) (1 - \rho(t+n\Delta t, T))+\cdots \\
 \cdots + \quad \sigma(t) \sigma(t+n\Delta t) (\rho(t, n\Delta t) - \rho(t, n\Delta t-T))
\end{aligned}}
{
\begin{aligned}
(\mu(t) - \mu(t+n\Delta t))^2 + \sigma^2(t) + \sigma^2(t+n\Delta t)+\cdots \\
 \cdots + \quad 2 \rho(t, n\Delta t-T) \sigma(t) \sigma(t+n\Delta t)
\end{aligned}}.
\end{equation}
\qed
\end{proof}
\begin{lemma}[Convexity of the $\mathtt{MSE} $  for the $\mathtt{BLEND}$ Operator]
\label{problem}
The Mean Squared Error ($\mathtt{MSE}$) for the $\mathtt{BLEND}$ Persistence Operator is convex with respect to $\lambda^{\circlearrowleft}_{\mathtt{BLEND}}$ for any cyclostationary or stationary process with finite variance and mean.
\end{lemma}
\begin{proof}
The $\mathtt{MSE}$ can be expressed as a quadratic function of $\lambda^{\circlearrowleft}_{\mathtt{BLEND}}$ with $\mathtt{MSE}(\lambda^{\circlearrowleft}_{\mathtt{BLEND}}) = A (\lambda^{\circlearrowleft}_{\mathtt{BLEND}})^2 - 2B \lambda^{\circlearrowleft}_{\mathtt{BLEND}} + C,$
where
\begin{align*}
&A = (\mu(t) - \mu(t+n\Delta t))^2 + \sigma^2(t) +  \\ & \sigma^2(t+n\Delta t) - 2 \rho(t, n\Delta t-T) \sigma(t) \sigma(t+n\Delta t).
\end{align*}
The second derivative is $\frac{d^2}{d(\lambda^{\circlearrowleft}_{\mathtt{BLEND}})^2} \mathtt{MSE} = 2A.$
Since $A$ consists of squared terms and variances, it satisfies $A \geq 0$ under standard conditions (correlation different than 1 and process not constant). Thus, $\mathtt{MSE}(\lambda^{\circlearrowleft}_{\mathtt{BLEND}})$ is convex, and has a unique minimizer.
\qed
\end{proof}
\begin{lemma}[Boundedness and Stability of $\lambda^{\circlearrowleft}_{\mathtt{BLEND}}$]
\label{problem2}
For cyclostationary processes, $\lambda^{\circlearrowleft}_{\mathtt{BLEND}}$ is generally bounded within $[0,1]$, ensuring convexity. However, instabilities may arise when the denominator approaches zero, particularly due to interactions involving $\rho(t, n\Delta t+T)$, $\rho(t, n\Delta t - T)$, $\sigma(t)$, $\sigma(t+n\Delta t)$, $\mu(t)$, and $\mu(t+n\Delta t)$. These instabilities are more likely when periodic variations in mean and variance are weak.
\end{lemma}
\begin{proof}
The numerator of $\lambda^{\circlearrowleft}_{\mathtt{BLEND}}$ consists of squared differences and variance-weighted correlations, ensuring it remains bounded. The denominator,  
\begin{align*}
&(\mu(t) - \mu(t+n\Delta t))^2 + \sigma^2(t) + \sigma^2(t+n\Delta t) - \\ &2 \rho(t, n\Delta t - T) \sigma(t) \sigma(t+n\Delta t),
\end{align*}
is typically positive when periodic variations are strong. This ensures that $\lambda^{\circlearrowleft}_{\mathtt{BLEND}} \geq 0$.
To guarantee $\lambda^{\circlearrowleft}_{\mathtt{BLEND}} \leq 1$, the numerator must not exceed the denominator. This holds when correlations remain moderate and periodic variations in $\mu(t)$ and $\sigma(t)$ are significant. However, if $\rho(t, n\Delta t - T) \approx 1$ and mean differences vanish, the denominator may become small, potentially allowing $\lambda^{\circlearrowleft}_{\mathtt{BLEND}} > 1$. This occurs when the process becomes nearly stationary.
Thus, $\lambda^{\circlearrowleft}_{\mathtt{BLEND}}$ remains stable and convexly bounded for most cyclostationary processes but may exceed 1 when periodic structures weaken significantly.
\qed
\end{proof}

\begin{corollary}[Stationary Case]
Under stationary assumptions, where:
\begin{align*}
\mu(t) =& \mu(t + n\Delta t) = \mu, \\ \sigma(t) = \sigma(t + n\Delta t) =& \sigma, \\ \rho(t, n\Delta t) =& \rho(n\Delta t),
\end{align*}
the formula for $\lambda$ simplifies to:
\begin{equation}
\lambda_{\mathtt{BLEND}} = \frac{\rho(T) - \rho(n\Delta t) + \rho(n\Delta t - T) - 1}{2(\rho(n\Delta t - T) - 1)}.
\end{equation}
\end{corollary}
This derivation highlights the general applicability of the $\mathtt{BLEND}$ Persistence operator while addressing both cyclostationary and stationary cases with mathematical rigor.

\section{Properties of \texttt{BLEND} Operator Family}
\label{annexe:proprietes}
The family of operators $\mathtt{P}_{\mathtt{BLEND}}$, $\mathtt{P}^{\circlearrowleft}_{\mathtt{BLEND}}$, and $\tilde{\mathtt{P}}^{\circlearrowleft}_{\mathtt{BLEND}}$ is defined by the general form:
\begin{equation}
\mathtt{P}_{\mathtt{BLEND}}(I)(t) = \lambda I(t - T) + (1 - \lambda) I(t),
\label{add}
\end{equation}
where $\lambda \in [0, 1]$ is the combination coefficient, $T$ represents the periodicity, and $I(t)$ is the observed signal. This formulation endows the operator family with several key mathematical properties.
\begin{theorem}[Linearity]
The operator $\mathtt{P}_{\mathtt{BLEND}}$ is linear. For scalars $a, b \in \mathbb{R}$ and signals $I_1, I_2$, the operator satisfies:
\begin{equation}
\mathtt{P}_{\mathtt{BLEND}}(aI_1 + bI_2)(t) = a\mathtt{P}_{\mathtt{BLEND}}(I_1)(t) + b\mathtt{P}_{\mathtt{BLEND}}(I_2)(t).
\end{equation}
Moreover, $\mathtt{P}_{\mathtt{BLEND}}$ is distributive over addition:
\begin{equation}
\mathtt{P}_{\mathtt{BLEND}}(I_1 + I_2)(t) = \mathtt{P}_{\mathtt{BLEND}}(I_1)(t) + \mathtt{P}_{\mathtt{BLEND}}(I_2)(t).
\end{equation}
\end{theorem}
\begin{proof}
The linearity follows directly from the definition of $\mathtt{P}_{\mathtt{BLEND}}$ as a weighted sum:
\begin{align*}
\mathtt{P}_{\mathtt{BLEND}}(aI_1 + bI_2)(t) =& \lambda (aI_1(t - T) + bI_2(t - T)) + \\ &(1 - \lambda)(aI_1(t) + bI_2(t)).
\end{align*}
Rearranging terms confirms the property:
\begin{align*}
\mathtt{P}_{\mathtt{BLEND}}(aI_1 + bI_2)(t) =& a\left[\lambda I_1(t - T) + (1 - \lambda)I_1(t)\right] + \\& b\left[\lambda I_2(t - T) + (1 - \lambda)I_2(t)\right],
\end{align*}
which implies $\mathtt{P}_{\mathtt{BLEND}}(aI_1 + bI_2)(t) = a\mathtt{P}_{\mathtt{BLEND}}(I_1)(t) + b\mathtt{P}_{\mathtt{BLEND}}(I_2)(t)$. Distributivity is a special case where $a = b = 1$. \qed
\end{proof}
\begin{theorem}[Flatness for Constant Signals]
For a constant signal $I(t) = c$, the operator is flat, returning:
\begin{equation}
\mathtt{P}_{\mathtt{BLEND}}(c)(t) = c.
\end{equation}
\end{theorem}
\begin{proof}
Substituting $I(t) = c$ into the definition \ref{add}, it comes:
\begin{equation}
\mathtt{P}_{\mathtt{BLEND}}(c)(t) = \lambda c + (1 - \lambda)c = c,
\end{equation}
proving the flatness of the operator for constant signals.
\qed
\end{proof}
\begin{corollary}[Shift for Affine Signals]
For an affine signal $I(t) = at + b$, the operator introduces a shift proportional to $-\lambda a T$, given by:
\begin{equation}
\mathtt{P}_{\mathtt{BLEND}}(at + b)(t) = at + b - \lambda a T.
\end{equation}
\end{corollary}
\begin{remark}
The operator interpolates between the simple Persistence $\mathtt{P}(I)(t) = I(t)$ and the cyclic Persistence $\mathtt{P}^{\circlearrowleft}(I)(t) = I(t - T)$. By varying $\lambda$ between 0 and 1, $\mathtt{P}_{\mathtt{BLEND}}$ transitions smoothly between these two prediction models, capturing both the latest observation and the periodic behavior of the signal.
\end{remark}
\begin{theorem}[Stability]
The operator is stable when $\lambda \in [0, 1]$, as predictions are bounded by the values of $I(t)$ and $I(t - T)$. Specifically:
\begin{equation}
\min(I(t), I(t - T)) \leq \mathtt{P}_{\mathtt{BLEND}}(I)(t) \leq \max(I(t), I(t - T)).
\end{equation}
\end{theorem}
\begin{proof}
Since $\lambda \in [0, 1]$, the weighted sum $\lambda I(t - T) + (1 - \lambda) I(t)$ lies within the convex hull of $I(t)$ and $I(t - T)$.
\qed
\end{proof}
\begin{remark}
By minimizing the Mean Squared Error ($\mathtt{MSE}$) for the choice of $\lambda$, $\mathtt{P}_{\mathtt{BLEND}}$ ensures optimal predictions in the least-squares sense. The dynamic adjustment of $\lambda$ based on correlations or coefficients of variation allows the operator to adapt to the variability of the signal.
\end{remark}
\begin{remark}
The operator's robustness and versatility make it particularly suited for noisy periodic signals, where a balanced combination of simple and cyclic Persistence is necessary. However, care must be taken when $\lambda$ approaches instability due to near-perfect correlations or insufficient variability in the signal.
\end{remark}

\bibliographystyle{ieeetr}

\bibliography{Bib} 

\end{document}